\documentclass[aps,twocolumn,nofootinbib,preprintnumbers,superscriptaddress]{revtex4-2}

\usepackage{color}
\usepackage{bbm}
\usepackage{tikz}
\usepackage{braket}
\usepackage{comment}
\usepackage{subcaption}
\usepackage[
pagebackref=false,
colorlinks=true,
linkcolor=blue,
urlcolor=blue,
filecolor=black,
citecolor=red,
pdfstartview=FitV,
pdftitle={},
pdfauthor={},
pdfsubject={},
pdfkeywords={},
pdfpagemode=None,
bookmarksopen=true
]{hyperref}

\usepackage[normalem]{ulem}
\usepackage{amsmath, amssymb, bm}
\usepackage{enumerate}
\usepackage{amsfonts}
\usepackage{epsfig}
\usepackage{mathbbol}
\usepackage{mathrsfs}
\usepackage{amsthm}

\newtheorem{thm}{Theorem}

\usetikzlibrary{shapes}
\usetikzlibrary{decorations.pathmorphing}
\tikzset{snake it/.style={decorate, decoration=snake}}

\newcommand\scalemath[2]{\scalebox{#1}{\mbox{\ensuremath{\displaystyle #2}}}}

\newcommand{\tr}{\textrm{Tr}}

\newcommand{\mA}{\mathcal{A}}
\newcommand{\mB}{\mathcal{B}}

\newcommand{\mD}{\mathcal{D}}
\newcommand{\g}{\gamma}

\begin{document}
\title{Petz map recovery for long-range entangled quantum many-body states}
\author{Yangrui Hu}
\affiliation{Perimeter Institute for Theoretical Physics, Waterloo, Ontario, Canada N2L2Y5}
\author{Yijian Zou}
\email{yzou@perimeterinstitute.ca}
\affiliation{Perimeter Institute for Theoretical Physics, Waterloo, Ontario, Canada N2L2Y5}

\begin{abstract}
Given a tripartite quantum state on $A,B,C$ and the erasure channel on $C$, the rotated Petz map is a recovery channel that acts on $B$ to recover the erased quantum information. The infidelity of the best recovery is upper-bounded by the conditional mutual information (CMI). In this work, we study the infidelity of the rotated Petz map on several physically-relevant long-range entangled quantum states. Specifically, we study three classes of quantum phases: (i) steady states of measurement-induced phase transitions, (ii) critical ground state under local measurements, and (iii) chiral states under local measurements. We find that the averaged infidelity of the Petz map recovery sharply distinguishes the three classes: (i) and (ii) are distinguished by the scaling of the infidelity with CMI and (iii) is characterized by an asymmetry of the fidelity with the rotation parameter. We also study Petz map recovery for topological order and find an operational interpretation of the topological entanglement entropy. Our result indicates that the recovery fidelity of the Petz map is a useful diagnostic of quantum phases of matter.  
\end{abstract}

\maketitle

\section{Introduction}
Quantum many-body entanglement provides extremely useful insights into quantum phases and phase transitions. The notion of long-range entanglement \cite{Chen_2010,Wen_Zoo_2017} underlies fundamental properties of nontrivial quantum phases of matter, such as topological order \cite{Kitaev_2006,Levin_2006} and quantum critical points \cite{Vidal_2003,Calabrese_2004}. Moreover, recent studies have demonstrated quantum phases which only manifest themselves in their entanglement properties. One notable example is the measurement-induced phase transitions (MIPT) in monitored quantum circuits \cite{Jian:2019mny,Choi:2019nhg,Li_2018,Li_2019,Skinner:2018tjl,Cao_2019}, which display distinct entanglement phases upon tuning the measurement rate.

While previous characterizations of quantum phases mostly focus on bipartite entanglement, multipartite entanglement also constitutes an important part of the patterns of long-range entanglement. Indeed, as opposed to bipartite entanglement where the entanglement entropy is the only measure, multipartite entanglement of a quantum many-body state can be characterized in exponentially many ways \cite{Walter_2016}. Different quantities characterize different aspects of many-body entanglement. A number of useful quantities such as negativity \cite{Vidal_2002} and Markov gap \cite{Hayden_2021} have been used to reveal universal properties of quantum phases, including both ground-state phases \cite{Marcovitch_2009,Bayat_2010,Kusuki_2019,Zou_2021,Siva_2022,Liu_2022,Liu_2024,Berthiere:2023bwn} and MIPTs \cite{Sang_2021,avakian2024longrangemultipartiteentanglementnear,Pavi_2023}. 

One important property of entanglement entropy is its operational interpretation in terms of the entanglement cost \cite{nielsen2010quantum}. Such interpretations are generally not available for multipartite entanglement. This raises the question of whether any operational task can inspire a quantity that characterizes multipartite entanglement and whether this quantity can be useful for identifying quantum phases of matter.

In this work, we will focus on the operational task of reversing an erasure operation. Given a tripartite state $\rho_{ABC}$ and the erasure channel on $C$, what is the best recovery channel acting on $B$ that recovers the erased information? This problem has been considered by Petz \cite{Petz:1986tvy,Petz:1988usv}, who proved that perfect recovery is possible if and only if the conditional mutual information (CMI) $I(A:C|B)$ vanishes, and constructed an explicit recovery map known as the Petz map. In a recent seminal paper \cite{Junge:2015lmb}, the Petz map is generalized to approximate recovery with a rotation parameter $t$, and the CMI is proven to provide a lower bound on the recovery infidelity. In quantum many-body systems, the Petz map has been a useful theoretical tool to prove the stability of quantum states against decoherence \cite{sang2024stabilitymixedstatequantumphases}. However, quantitative features of the recovery fidelity have yet to be explored for characterizing quantum phases of matter. In this work, we study the fidelity of the Petz map for various examples of measurement-induced phases and topological phases. We find that the fidelity has sharp features that distinguish these phases, extending beyond the insights provided by the CMI. Thus, the fidelity of the Petz map recovery offers a useful operational characterization of quantum phases.

Specifically, we study three quantitative aspects of the Petz map fidelity: the rotational parameter $t$ for the best recovery, the scaling of infidelity with the CMI, and the symmetry of fidelity with respect to $t$. For MIPTs with Clifford random, Haar random, and $U(1)$-symmetric Haar random unitaries, we observe a universal behavior of the Petz map infidelity\textemdash linear in the CMI\textemdash despite varying critical points and critical exponents. We explain this behavior using the mapping to statistical mechanical model \cite{Jian:2019mny,Li_2021v2}. In sharp contrast, for critical ground states with or without local measurements, we observe a quadratic relation between the infidelity and the CMI. We further examine gapless edge states of a two-dimensional chiral topological order and find that the quadratic relation still holds for the best fidelity, signaling conformal symmetry. However, the fidelity exhibits asymmetry with respect to $t$, which is a sharp characterization of chirality. The behavior persists for the chiral state under local measurements, indicating that chirality is protected under such measurements. Finally, we apply the Petz map to the bulk of a topological order and reproduce a previous result of the non-negativity of spurious topological entanglement entropy from an operational perspective. Our main results are summarized in Tab.~\ref{tab:main}.
\begin{table*}[t]
    \centering
    \begin{tabular}{|c|c|c|c|}
    \hline
         &  Best recovery & Best infidelity v.s. CMI & $t$ dependence \\ \hline
      Stabilizer states &       any $t$        &    $-\log F = \mathrm{CMI}/2$            &   uniform             \\ \hline
      Haar MIPT    &          $t=0$      &      $-\overline{\log F} \propto \overline{\mathrm{CMI}}$            & symmetric            \\ \hline
      Haar MIPT with $U(1)$ symmetry   &    $t=0$            &    $-\overline{\log F} \propto \overline{\mathrm{CMI}}$              &       symmetric        \\ \hline
      Critical ground state   &         $t=0$       &     $-\log F \propto \mathrm{CMI}^2$ \cite{Vardhan:2023pnm}             &      symmetric          \\ \hline
      Critical ground state with measurements  &      $t=0$          &   $-\overline{\log F} \propto \overline{\mathrm{CMI}}^2$              &     symmetric          \\ \hline
      Chiral edge state  &       $t\neq 0$         &    $-\log F \propto \mathrm{CMI}^2$              &    asymmetric            \\ \hline
      Chiral edge state with measurements &   $t\neq 0$           &    $-\overline{\log F} \propto \overline{\mathrm{CMI}}^2$              &        asymmetric         \\ \hline
    \end{tabular}
    \caption{Petz map infidelity in different quantum many-body systems}
    \label{tab:main}
\end{table*}

The rest of the paper is organized as follows. In Sec.~\ref{sec:Petz} we review the Petz map and its generalization with the rotation parameter $t$. Sec.~\ref{sec:MIPT} discusses the Petz map fidelity for various MIPTs, showing its linear relation with the CMI both analytically and numerically. In Sec.~\ref{sec:CFT} we consider the Petz map fidelity for critical states and the chiral edge states under measurements. Sec.~\ref{sec:TO} is devoted to the Petz map fidelity in a topologically order system. We conclude with a summary and future directions in Sec.~\ref{sec:conc} followed by two appendices that review the numerical techniques used for the main results, including the Gaussian fermion formalism and exact diagonalization.

\section{Petz map}
\label{sec:Petz}
Quantum operations are in general irreversible. Given two states $\rho,\sigma$ and a quantum channel $\mathcal{N}$, their distinguishability cannot increase under the operation, which can be seen in the monotonicity of the relative entropy
\begin{equation}
\label{eq:relS_monotonicity}
    S(\rho||\sigma) ~\geq~ S(\mathcal{N}(\rho)||\mathcal{N}(\sigma))~,
\end{equation}
where $S(\rho||\sigma):=\tr(\rho\log \rho - \rho \log \sigma)$ is the relative entropy between the two states\footnote{Note that throughout this work we will take the logarithm to be base $2$.}. Petz~\cite{Petz:1986tvy,Petz:1988usv} showed that the channel $\mathcal{N}$ is reversible on these two states if and only if the equal sign holds in the inequality, in which case the reverse channel $\mathcal{D}_{\sigma,0}$ has an explicit form known as the Petz map
\begin{equation}
    \mathcal{D}_{\sigma,0}(X) ~=~ \sigma^{\frac{1}{2}}\mathcal{N}^{*}(\mathcal{N}(\sigma)^{-\frac{1}{2}} X \mathcal{N}(\sigma)^{-\frac{1}{2}})\sigma^{\frac{1}{2}}~,
\end{equation}
where $\mathcal{N}^{*}$ is the dual channel. It is clear that $\mathcal{D}_{\sigma,0}(\mathcal{N}(\sigma)) = \sigma$ reverses the quantum channel on $\sigma$, as the dual channel is always unital, that is, $\mathcal{N}^{*}(\mathbb{1}) = \mathbb{1}$. If the equality sign holds in Eq.~\eqref{eq:relS_monotonicity}, then the Petz map also reverses the quantum channel on $\rho$, namely, $\mathcal{D}_{\sigma,0}(\mathcal{N}(\rho)) = \rho$.

Furthermore, if the equality sign only approximately holds in Eq.~\eqref{eq:relS_monotonicity}, it has been shown~\cite{Junge:2015lmb,Wilde_2015} that there exists an approximate recovery channel, which extends the Petz map to a family of rotated Petz maps
\begin{equation}
    \mathcal{D}_{\sigma,t}(X) = \sigma^{\frac{1+it}{2}}\mathcal{N}^{*}(\mathcal{N}(\sigma)^{\frac{-1-it}{2}} X \mathcal{N}(\sigma)^{\frac{-1+it}{2}})\sigma^{\frac{1-it}{2}}~,
    \label{equ:rotated_Petz}
\end{equation}
where $t$ is a real number. 
The rotated Petz map still perfectly recovers $\sigma$, $\mathcal{D}_{\sigma,t}(\mathcal{N}(\sigma)) = \sigma$, and approximately recovers the state $\rho$. It has been shown that the best Petz map fidelity satisfies the following inequality \cite{Wilde_2015} 
\begin{equation}
\label{eq:fid_bound_0}
    \max_{t} F(\mathcal{D}_{\sigma,t}(\mathcal{N}(\rho)),\rho)~\geq~ 2^{-[S(\rho||\sigma) - S(\mathcal{N}(\rho)||\mathcal{N}(\sigma))]/2}~,
\end{equation}
where $F(\rho,\sigma):=||\sqrt{\rho} \sqrt{\sigma}||_1$ is the Uhlmann fidelity. Furthermore, there exists the twirled Petz map
\begin{equation}
\label{eq:twirled}
    \mathcal{D}_{\sigma}(X) ~=~ \int_{-\infty}^{\infty} dt\, \beta(t) \mathcal{D}_{\sigma,t}(X)
\end{equation}
where $\beta(t) = \frac{\pi}{2(1+\cosh(\pi t))}$ and the fidelity has the following lower bound
\begin{equation}
    F(\mathcal{D}_{\sigma}(\mathcal{N}(\rho)),\rho)~\geq~ 2^{-[S(\rho||\sigma) - S(\mathcal{N}(\rho)||\mathcal{N}(\sigma))]/2}~.
\end{equation}

Now we consider a tripartite state $\rho_{ABC}$ supported on three parties $A$, $B$, and $C$. The state forms an approximate quantum Markov chain \cite{Wilde_2015} if the conditional mutual information (CMI) $I(A:C|B)$ is small, where $I(A:C|B) = S_{AB} + S_{BC} - S_B - S_{ABC}$ and $S$ is the Von-Neumann entropy. The rotated Petz map Eq.~\eqref{equ:rotated_Petz} can be applied to approximately recover the erasure channel $\mathcal{N}$ on $C$, where
\begin{equation}
    \mathcal{N}(\rho_{ABC}) ~=~ \rho_{AB}~,
\end{equation}
$\rho:=\rho_{ABC}$ and the reference state is $\sigma:=\rho_{A}\otimes \rho_{BC}$. With such a choice of $\sigma$, the relative entropy difference becomes the CMI, $S(\rho||\sigma) - S(\mathcal{N}(\rho)||\mathcal{N}(\sigma)) = I(A:C|B)$. If the state $\rho_{ABC}$ is an approximate quantum Markov chain, then Eq.~\eqref{eq:fid_bound_0} ensures that the recovered state is closed to the original state $\rho_{ABC}$.

More concretely, the rotated Petz map for recovering the erasure channel on $C$ becomes
\begin{equation}
\label{eq:Petz_map_1}
    \mathcal{D}_{\rho_{A}\otimes \rho_{BC},t}(\rho_{AB}) ~=~ \rho^{\frac{1+it}{2}}_{BC} \rho^{\frac{-1-it}{2}}_B \rho_{AB}\, \rho^{\frac{-1+it}{2}}_B \rho^{\frac{1-it}{2}}_{BC}~.
\end{equation}
and Eq.~\eqref{eq:fid_bound_0} becomes
\begin{equation}
\label{eq:fid_bound_1}
   \max_t F(\tilde{\rho}_{ABC}(t), \rho_{ABC}) ~\geq~ 2^{-I(A:C|B)/2}~,
\end{equation}
where we have defined the shorthand notation $\tilde{\rho}_{ABC}(t):=\mathcal{D}_{\rho_{A}\otimes \rho_{BC},t}(\rho_{AB})$ for the recovered state. The rotated Petz map is perfect if and only if $I(A:C|B)=0$, in which case the state $\rho_{ABC}$ is an exact quantum Markov chain. In what follows, we will often use the shorthand notation $F_t := F(\tilde{\rho}_{ABC}(t), \rho_{ABC})$. One can also define the twirled Petz map by integrating $\mathcal{D}_{\rho_{A}\otimes \rho_{BC},t}$ as in Eq.~\eqref{eq:twirled}.

Although there are nice bounds such as Eq.~\eqref{eq:fid_bound_1}, it is less understood how the fidelity of the Petz map varies with $t$ and how tight the bound is for physically-relevant quantum states. In the following sections, we explore these questions using specific examples in quantum many-body systems. Given a quantum many-body state or an ensemble of quantum many-body states, we can obtain a series of mixed states $\rho_{ABC}$ whose CMI's scale to zero by varying the sizes of $A,B,C$.  We then study how the fidelity of the Petz map changes with $t$ and the CMI. Our work is a generalization of a previous work \cite{Vardhan:2023pnm} which specifically focused on critical ground states.

\section{Petz map for measurement-induced phase transitions}
\label{sec:MIPT}
\subsection{Measurement-induced phase transitions}
Measurement-induced phase transition (MIPT) is an entanglement phase transition in monitored quantum circuits. The circuit consists of $L$ spins that undergo repeated layers of local random unitary dynamics and projective measurements. Schematically, 
\begin{equation}
    |\psi(\tau)\rangle ~\propto~ \left(\prod_{l=1}^{\tau} P_l U_l \right)|\psi(0)\rangle~,
\end{equation}
as depicted in Fig.~\ref{fig:mipt}. 
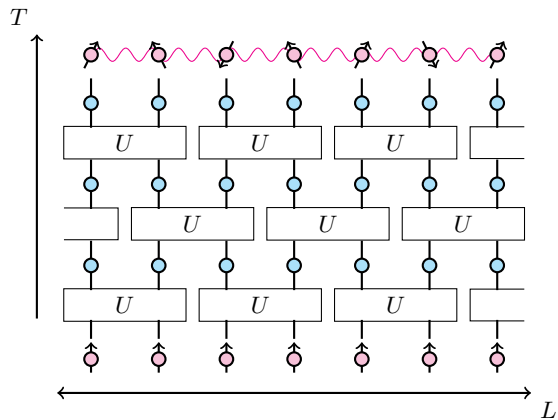
\begin{figure}[t]
    \centering
    \begin{tikzpicture}[scale=0.9]
    %%%%%%%%%%%%%%%%%%%%%%%%%%%%%%%%%%%%%
    \draw[thick, ->] (-2.5,-0.8-0.2) -- (-2.5,-0.8+0.25);
    \draw[thick, ->] (-1.5,-0.8-0.2) -- (-1.5,-0.8+0.25);
    \draw[thick, ->] (-0.5,-0.8-0.2) -- (-0.5,-0.8+0.25);
    \draw[thick, ->] (0.5,-0.8-0.2) -- (0.5,-0.8+0.25);
    \draw[thick, ->] (1.5,-0.8-0.2) -- (1.5,-0.8+0.25);
    \draw[thick, ->] (2.5,-0.8-0.2) -- (2.5,-0.8+0.25);
    \draw[thick, ->] (3.5,-0.8-0.2) -- (3.5,-0.8+0.25);
    \filldraw[thick,fill=magenta!30!white] (-2.5,-0.8) circle (0.1);
    \filldraw[thick,fill=magenta!30!white] (-1.5,-0.8) circle (0.1);
    \filldraw[thick,fill=magenta!30!white] (-0.5,-0.8) circle (0.1);
    \filldraw[thick,fill=magenta!30!white] (0.5,-0.8) circle (0.1);
    \filldraw[thick,fill=magenta!30!white] (1.5,-0.8) circle (0.1);
    \filldraw[thick,fill=magenta!30!white] (2.5,-0.8) circle (0.1);
    \filldraw[thick,fill=magenta!30!white] (3.5,-0.8) circle (0.1);
    \draw[thick] (0.5,-0.22) -- (0.5,-0.5);
    \draw[thick] (-2.5,-0.22) -- (-2.5,-0.5);
    \draw[thick] (-1.5,-0.22) -- (-1.5,-0.5);
    \draw[thick] (-0.5,-0.22) -- (-0.5,-0.5);
    \draw[thick] (1.5,-0.22) -- (1.5,-0.5);
    \draw[thick] (2.5,-0.22) -- (2.5,-0.5);
    \draw[thick] (3.5,-0.22) -- (3.5,-0.5);
    %%%%%%%%%%%%%%%%%%%%%%%%%%%%%%%%%%%%%
    \node[draw, shape=rectangle, minimum width=5em] (U1) at (0,0){$U$};
    \node[draw, shape=rectangle, minimum width=5em] (U1) at (-2,0){$U$};
    \node[draw, shape=rectangle, minimum width=5em] (U1) at (2,0){$U$};
    \draw (3.9,0.235) -- (3.1,0.235) -- (3.1,-0.235) -- (3.9,-0.235);
    \draw[thick] (-2.5,0.22) -- (-2.5,0.95);
    \draw[thick] (-1.5,0.22) -- (-1.5,0.95);
    \draw[thick] (-0.5,0.22) -- (-0.5,0.95);
    \draw[thick] (0.5,0.22) -- (0.5,0.95);
    \draw[thick] (1.5,0.22) -- (1.5,0.95);
    \draw[thick] (2.5,0.22) -- (2.5,0.95);
    \draw[thick] (3.5,0.22) -- (3.5,0.95);
    \filldraw[thick,fill=cyan!30!white] (-2.5,0.22+0.73/2) circle (0.1);
    \filldraw[thick,fill=cyan!30!white] (-1.5,0.22+0.73/2) circle (0.1);
    \filldraw[thick,fill=cyan!30!white] (0.5,0.22+0.73/2) circle (0.1);
    \filldraw[thick,fill=cyan!30!white] (-0.5,0.22+0.73/2) circle (0.1);
    \filldraw[thick,fill=cyan!30!white] (1.5,0.22+0.73/2) circle (0.1);
    \filldraw[thick,fill=cyan!30!white] (2.5,0.22+0.73/2) circle (0.1);
    \filldraw[thick,fill=cyan!30!white] (3.5,0.22+0.73/2) circle (0.1);
    %%%%%%%%%%%%%%%%%%%%%%%%%%%%%%%%%%%%%
    \node[draw, shape=rectangle, minimum width=5em] (U1) at (1,1.2){$U$};
    \node[draw, shape=rectangle, minimum width=5em] (U1) at (-1,1.2){$U$};
    \node[draw, shape=rectangle, minimum width=5em] (U1) at (3,1.2){$U$};
    \draw (-2.9,0.235+1.2) -- (-2.1,0.235+1.2) -- (-2.1,-0.235+1.2) -- (-2.9,-0.235+1.2);
    \draw[thick] (0.5,0.22+1.2) -- (0.5,0.95+1.2);
    \draw[thick] (-2.5,0.22+1.2) -- (-2.5,0.95+1.2);
    \draw[thick] (-1.5,0.22+1.2) -- (-1.5,0.95+1.2);
    \draw[thick] (-0.5,0.22+1.2) -- (-0.5,0.95+1.2);
    \draw[thick] (1.5,0.22+1.2) -- (1.5,0.95+1.2);
    \draw[thick] (2.5,0.22+1.2) -- (2.5,0.95+1.2);
    \draw[thick] (3.5,0.22+1.2) -- (3.5,0.95+1.2);
    \filldraw[thick,fill=cyan!30!white] (-2.5,0.22+0.73/2+1.2) circle (0.1);
    \filldraw[thick,fill=cyan!30!white] (-1.5,0.22+0.73/2+1.2) circle (0.1);
    \filldraw[thick,fill=cyan!30!white] (0.5,0.22+0.73/2+1.2) circle (0.1);
    \filldraw[thick,fill=cyan!30!white] (-0.5,0.22+0.73/2+1.2) circle (0.1);
    \filldraw[thick,fill=cyan!30!white] (1.5,0.22+0.73/2+1.2) circle (0.1);
    \filldraw[thick,fill=cyan!30!white] (2.5,0.22+0.73/2+1.2) circle (0.1);
    \filldraw[thick,fill=cyan!30!white] (3.5,0.22+0.73/2+1.2) circle (0.1);
    %%%%%%%%%%%%%%%%%%%%%%%%%%%%%%%%%%%%%
    \node[draw, shape=rectangle, minimum width=5em] (U1) at (0,2.4){$U$};
    \node[draw, shape=rectangle, minimum width=5em] (U1) at (-2,2.4){$U$};
    \node[draw, shape=rectangle, minimum width=5em] (U1) at (2,2.4){$U$};
    \draw[thick] (0.5,0.22+2.4) -- (0.5,0.95+2.4);
    \draw[thick] (-2.5,0.22+2.4) -- (-2.5,0.95+2.4);
    \draw[thick] (-1.5,0.22+2.4) -- (-1.5,0.95+2.4);
    \draw[thick] (-0.5,0.22+2.4) -- (-0.5,0.95+2.4);
    \draw[thick] (1.5,0.22+2.4) -- (1.5,0.95+2.4);
    \draw[thick] (2.5,0.22+2.4) -- (2.5,0.95+2.4);
    \draw[thick] (3.5,0.22+2.4) -- (3.5,0.95+2.4);
    \draw (3.9,0.235+2.4) -- (3.1,0.235+2.4) -- (3.1,-0.235+2.4) -- (3.9,-0.235+2.4);
    \filldraw[thick,fill=cyan!30!white] (-2.5,0.22+0.73/2+2.4) circle (0.1);
    \filldraw[thick,fill=cyan!30!white] (-1.5,0.22+0.73/2+2.4) circle (0.1);
    \filldraw[thick,fill=cyan!30!white] (0.5,0.22+0.73/2+2.4) circle (0.1);
    \filldraw[thick,fill=cyan!30!white] (-0.5,0.22+0.73/2+2.4) circle (0.1);
    \filldraw[thick,fill=cyan!30!white] (1.5,0.22+0.73/2+2.4) circle (0.1);
    \filldraw[thick,fill=cyan!30!white] (2.5,0.22+0.73/2+2.4) circle (0.1);
    \filldraw[thick,fill=cyan!30!white] (3.5,0.22+0.73/2+2.4) circle (0.1);
    %%%%%%%%%%%%%%%%%%%%%%%%%%%%%%%%%%%%%
    \path [draw=magenta, snake it]
    (-2.5,3.7) -- (-1,3.7) -- (1,3.7) -- (3.5,3.7);
    \draw[thick, ->] (-2.5-0.1,3.7-0.2) -- (-2.5+0.1,3.7+0.2);
    \draw[thick, ->] (-1.5+0.1,3.7-0.2) -- (-1.5-0.1,3.7+0.2);
    \draw[thick, <-] (-0.5-0.1,3.7-0.2) -- (-0.5+0.1,3.7+0.2);
    \draw[thick, ->] (0.5+0.1,3.7-0.2) -- (0.5-0.1,3.7+0.2);
    \draw[thick, ->] (1.5-0.1,3.7-0.2) -- (1.5+0.1,3.7+0.2);
    \draw[thick, <-] (2.5+0.1,3.7-0.2) -- (2.5-0.1,3.7+0.2);
    \draw[thick, ->] (3.5-0.1,3.7-0.2) -- (3.5+0.1,3.7+0.2);
    \filldraw[thick,fill=magenta!30!white] (-2.5,3.7) circle (0.1);
    \filldraw[thick,fill=magenta!30!white] (-1.5,3.7) circle (0.1);
    \filldraw[thick,fill=magenta!30!white] (-0.5,3.7) circle (0.1);
    \filldraw[thick,fill=magenta!30!white] (0.5,3.7) circle (0.1);
    \filldraw[thick,fill=magenta!30!white] (1.5,3.7) circle (0.1);
    \filldraw[thick,fill=magenta!30!white] (2.5,3.7) circle (0.1);
    \filldraw[thick,fill=magenta!30!white] (3.5,3.7) circle (0.1);
    %%%%%%%%%%%%%%%%%%%%%%%%%%%%%%%%%%%%%
    \draw[thick,<->] (-3,-1.3) -- (4,-1.3) node[anchor=north west] {$L$};
    \draw[thick,->] (-3.3,-0.2) -- (-3.3,4) node[anchor=south east] {$T$};
    \end{tikzpicture}
\caption{Schematic diagram of the circuit. Rectangles represent random unitary gates acting on adjacent two qubits. Cyan circles denote local projective measurements with a measurement probability $p$, or identity map with probability $1-p$. The size of the whole spin chain is $L$ and the periodic boundary condition is adopted. }
\label{fig:mipt}
\end{figure}
The unitary evolution is taken as
\begin{equation}
    U_l ~=~ \bigotimes_{i\in \mathrm{even/odd}} U^{[i,i+1]}_{l}~,
\end{equation}
where $l$ denotes the number of layers of the circuit and $i$ alternates between even and odd numbers for different layers. Each unitary gate is chosen randomly within a distribution, such as a random two-qubit Clifford gate or a Haar random $U(4)$ gate. The local projective measurements can be chosen on the $Z$ basis and occur randomly with a probability $p$ at each site, known as the measurement rate. The projector reads
\begin{equation}
    P_l ~=~ \bigotimes_i P^{[i]}_{m_i}~,~ P^{[i]}_{\pm} ~=~ \frac{I \pm \sigma^{[i]}_z}{2} ~,
\end{equation}
where $i$ denotes the measured site and $m_i = \pm $ is the measurement outcome. Each measurement outcome is picked by the Born rule. Under the joint unitary-measurement dynamics, the states $\{|\psi(\tau)\rangle\}$ form an ensemble labeled by the random unitaries $U_l$ and measurement outcomes $P_l$. It takes depth $O(L)$ for the dynamics to reach equilibrium and we will choose $\tau=4L$ for the steady state ensemble $\{|\psi(\tau)\rangle\}$. 

The steady-state ensemble exhibits an entanglement phase transition if one varies the measurement rate $p$. Consider entanglement entropy of a subsystem $A$, averaged over the random unitaries and measurement outcomes,
\begin{equation}
    \bar{S}(A) ~=~ \mathbb{E}_{\{U_l\},\{P_l\}}\, S_A(|\psi(\tau)\rangle)~,
\end{equation}
where $S_A(|\psi\rangle) = -\tr(\rho_A \log \rho_A)$ is the entropy and $\rho_A = \tr_{\bar{A}}(|\psi\rangle\langle\psi|)$. 
Let us consider the extreme cases and analuze the phase transition in more detail \cite{Li_2018,Jian:2019mny}. If $p=0$, there is no measurement and the random unitary dynamics thermalizes the state to infinite tempeature quickly. As a result, any subsystem $A$ that is smaller than half of the total system in the steady state is close to a maximally mixed state and the averaged entanglement $\bar{S}(A) \propto L_A$. Such a volume law phase persists to small measurement rate until a phase transition occurs. If $p=1$, then at any time step the state is randomly projected onto a product state and $\bar{S}(A) = 0$ for any region. For large $p$ but $p<1$, most of the entanglement built by the unitary dynamics gets lost and only some short-range entanglement persists, thus $\bar{S}(A) = O(1)$ is bounded by a constant. This is known as the area law phase of the MIPT. 

There is critical point $p=p_c$ which separates the volume law phase and the area law phase. If $p<p_c$, the local measurements cannot extract extensive quantum information from the random unitary dynamics and $\bar{S}(A) \propto L_A$ follows the volume law. If $p>p_c$, then $\bar{S}(A) = O(1)$ becomes the area law. It is remarkable that the entropy is logarithmic at the critical point, namely $\bar{S}(A) = \alpha\log L_A + O(1)$, indicating the scale invariance. 

The critical measurement rate $p_c$ and the factor $\alpha$ depend on the ensemble of the random unitary. The phase transition occurs at a finite $p_c$ if the random unitary dynamics is sufficiently scrambling. For example, one can choose Haar random unitaries or random Clifford unitaries. For the former case, we call the phase transition Haar MIPT. For the latter case, we will call the phase transition stabilizer MIPT since a Clifford unitary maps stablizer states to stablizer states. The two classes are differed by $p_c$ and $\alpha$.

We consider the Petz map recovery Eq.~\eqref{eq:Petz_map_1} for the steady state ensemble of different MIPTs at criticality. Let $A,B,C$ be three contiguous intervals as in the following figure, 
\begin{equation}
\nonumber
    \includegraphics[width = 0.3\linewidth]{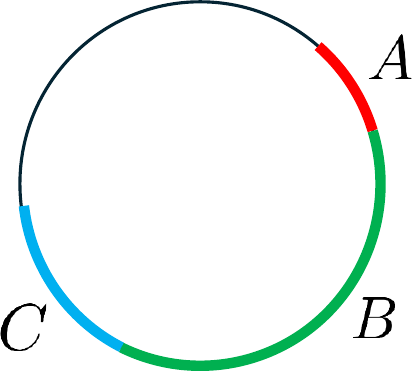}
\end{equation}
and the lengths of the intervals being $L_A$, $L_B$, and $L_C$. The averaged CMI over the ensemble of the steady states of MIPT at the critical point takes the following form due to the logarithmic entropy scaling,
\begin{equation}
\label{eq:CMI_CFT}
    \overline{I}(A:C|B) ~=~ \alpha\, \log \frac{1}{1-\eta}~,
\end{equation}
where
\begin{equation}
    \eta ~=~ \frac{\sin(\pi L_A/L) \sin(\pi L_C/L)}{\sin(\pi (L_A+L_B)/L) \sin(\pi (L_B+L_C)/L)}
\end{equation}
is the conformally-invariant cross-ratio. At small $\eta$, we find that
\begin{equation}
\label{eq:CMI_smalleta}
    \overline{I}(A:C|B) ~=~ \alpha\, \eta ~+~ O(\eta^2)~.
\end{equation}
In what follows, we will show that the best recovery fidelity of the Petz map can always be achieved at $t=0$ and that $-\overline{\log F_0} \propto \overline{I}(A:C|B)$ at small $\eta$. We provide analytical proof for the stabilizer MIPT. For Haar random MIPT, we both provide numerical evidence and an analytical explanation in terms of boundary conformal field theory.

\subsection{Petz map for Stabilizer MIPT}
For the stabilizer MIPT, the state at any time slice is a stabilizer state because each random unitary is Clifford. For a review of stabilizer states, see Refs.~\cite{gottesman1997stabilizercodesquantumerror,Li_2018}. We will show in Theorem \ref{thm1} that $-\log F_t = I(A:C|B)/2$ for any stabilizer density matrix $\rho_{ABC}$ and any $t$. It follows that $-\overline{\log F} = \overline{I}(A:C|B)/2$ for the ensemble of the steady states. 
\begin{thm}
\label{thm1}
For any stabilizer state $\rho_{ABC}$, the Petz map fidelity $F_t$ is independent of $t$ and saturates the lower bound in Eq.~\eqref{eq:fid_bound_1}, i.e. $F_{t} = 2^{-I(A:C|B)/2}$.
\end{thm}
\begin{proof}
Firstly, we note that a generic stabilizer density matrix $\rho_R$ on a region $R$ is proportional to a projector. Denote the stabilizer group on $R$ by $G_R$ and its generators by $\{g_1,g_2,\cdots g_{r_R}\}$, then 
\begin{equation}
\label{eq:rho_stab}
    \rho_R ~=~ \frac{1}{2^{|R|}}\prod_{i=1}^{r_R} (I+g_i)~,
\end{equation}
where $|R|$ is the number of spins in $R$. The density matrix satisfies $\rho^{2}_R = 2^{r_R-|R|}\rho_R$, thus it is proportional to a projector, whose any power equals itself. This implies that
\begin{equation}
\label{eq:rho_power_stab}
    \rho^{z}_R ~=~ 2^{(z-1)(r_R-|R|)}\rho_R
\end{equation}
for any complex number $z$. Substituting Eq.~\eqref{eq:rho_power_stab} into Eq.~\eqref{eq:Petz_map_1} we obtain
\begin{equation}
    \rho_{ABC}(t)~\propto~ \rho_{BC} \rho_B \rho_{AB} \rho_B \rho_{BC}~.
\end{equation}
It is then clear that $ \rho_{ABC}(t)$ is independent of $t$, thus $F_t$ is also independent of $t$. Since the stabilizers commute, the density matrices also commute, $[\rho_{B},\rho_{AB}] = 0,~[\rho_{BC},\rho_{AB}] = 0$. The above product specifies a new stabilizer state $\tilde{\rho}_{ABC}$ within the stabilizer group 
\begin{equation}
\label{eq:Gtilde}
    \tilde{G}_{ABC} ~=~ \{g_1 g_2|g_1 \in G_{AB}, g_2 \in G_{BC} \}~.
\end{equation} 
The above definition Eq.~\eqref{eq:Gtilde} of $\tilde{G}_{ABC}$ implies that
\begin{equation}
\label{eq:NGtilde}
    |\tilde{G}_{ABC}| ~=~ \frac{|G_{AB}| |G_{BC}|}{|G_{B}|} ~,
\end{equation}
since the only way to have $g_1g_2 = I$ is that $g_1 \in G_B$ and $g_2 = g^{-1}_1$. 
Here $|G_R|$ denotes the number of elements in the stabilizer group $G_R$ of the region $R$, which equals $2^{r_R}$. 
Note that $\tilde{G}_{ABC}$ is a normal subgroup of $G_{ABC}$.
The fidelity of such two stabilizer states $\rho_{ABC}$ and $\tilde{\rho}_{ABC}$ is given by
\begin{equation}
\label{eq:Ftstab}
    F_t ~=~ |G_{ABC}/\tilde{G}_{ABC}|^{-1/2} ~=~\sqrt{\frac{|\tilde{G}_{ABC}|}{|G_{ABC}|}}~.
\end{equation}
The above formula can be directly verified by plugging Eq.~\eqref{eq:rho_stab} into the definition of the fideliy $F(\rho,\tilde{\rho}) = \tr \sqrt{\sqrt{\rho}\tilde{\rho}\sqrt{\rho}}$. To compute the CMI, we note that the entropy of $\rho_R$ is given by
\begin{equation}
    S_{R} ~=~ |R| - \log |G_R|~.
\end{equation}
Thus,
\begin{equation}
\label{eq:CMIstab}
    I(A:C|B) ~=~ \log \frac{|G_B| | G_{ABC}|}{|G_{AB}||G_{BC}|} ~.
\end{equation}
Combining Eqs.~\eqref{eq:NGtilde},~\eqref{eq:Ftstab}, and~\eqref{eq:CMIstab}, we obtain
\begin{equation}
    I(A:C|B) ~=~ -2\log F_t
\end{equation}
as desired.
\end{proof}

\subsection{Petz map for Haar MIPT with or without symmetries}
For Haar MIPT, we consider the unitary to be either without symmetry or with a conserved $U(1)$ charge. In the Pauli $Z$ basis, the $U(1)$-symmetric Haar random unitary takes a block diagonal form $\mathrm{diag}(e^{i\theta_1}, V_{2\times 2}, e^{i\theta_2})$, where $\theta_1$ and $\theta_2$ are randomly chosen within $[0,2\pi)$ and $V_{2\times 2}$ is a Haar random unitary in $U(2)$. It is known that the two MIPTs have different critical measurement rates and different universality classes at the phase transitions. For Haar random MIPT $p_c \approx 0.168$ and $\alpha \approx 1.1$ \cite{Zabalo:2019sfl}, and for the $U(1)$-symmetric case $p_c\approx 0.105$ and $\alpha \approx 1.3$ \cite{Agrawal:2021ukw}. 

In the following, we consider the Petz map fidelity at small cross-ratios for these two MIPTs at criticality. We perform numerical simulations using exact diagonalization for up to $L=20$, see Appendix~\ref{appen:petz_via_diag} for details. For both cases, we observe that $-\overline{\log F_t}$ is minimized at $t=0$ and is symmetric with respect to $t\rightarrow -t$ (see Fig.~\ref{fig:Flambda_MIPT}). The symmetry can be explained by the average time-reversal symmetry: for every $\rho_{ABC}$ in the ensemble, there is a time-reversal conjugate $\rho^{*}_{ABC}$ with the same probability density. We also observe that the best infidelity of the Petz map recovery $-\overline{\log F_0} \approx 0.22 I(A:C|B)$ for both $U(1)$-symmetric and nonsymmetric cases, as shown in Fig.~\ref{fig:F0_CMI_mipt}. This is to be contrasted with the stabilizer case $-\overline{\log F_t} = I(A:C|B)/2$ which saturates the bound in Eq.~\eqref{eq:fid_bound_1} for all $t$. Yet, we still observe a linear relation between the average infidelity and the average CMI. This linear relation persists across all MIPTs we consider.

\begin{figure}[htp!]
    \centering
    \includegraphics[width=0.5\textwidth]{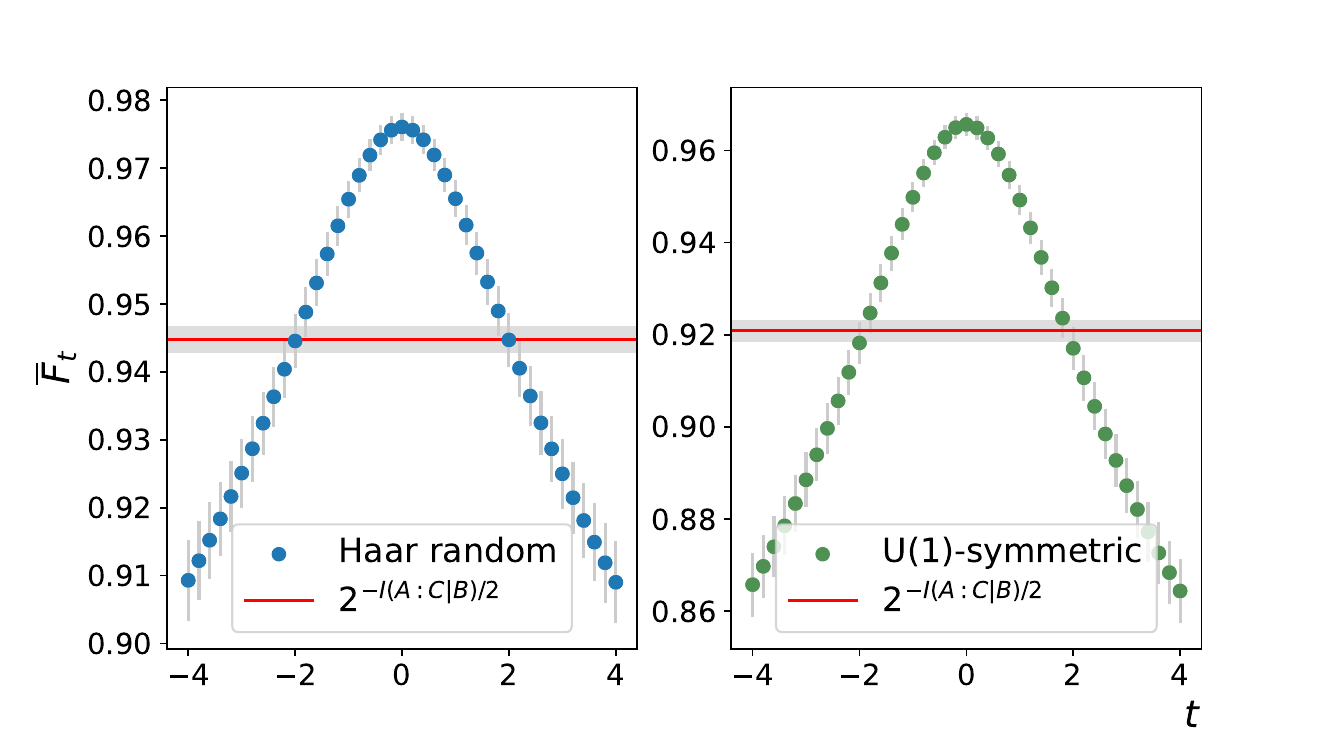}
    \caption{The relation between the averaged fidelity of rotated Petz map and parameter $t$. Here we take $L=20$, $L_A = 2$, $L_B = 8$, $L_C=2$ ($\eta = 0.095$) and average the results of $N=100$ simulations. As a consistency check of Eq.~(\ref{eq:fid_bound_1}), the quantity $2^{-I(A:C|B)/2}$ is also plotted in red. 
    Error bars are shown in gray.}
    \label{fig:Flambda_MIPT}
\end{figure}

\begin{figure}[htp!]
    \centering
    \includegraphics[width=0.5\textwidth]{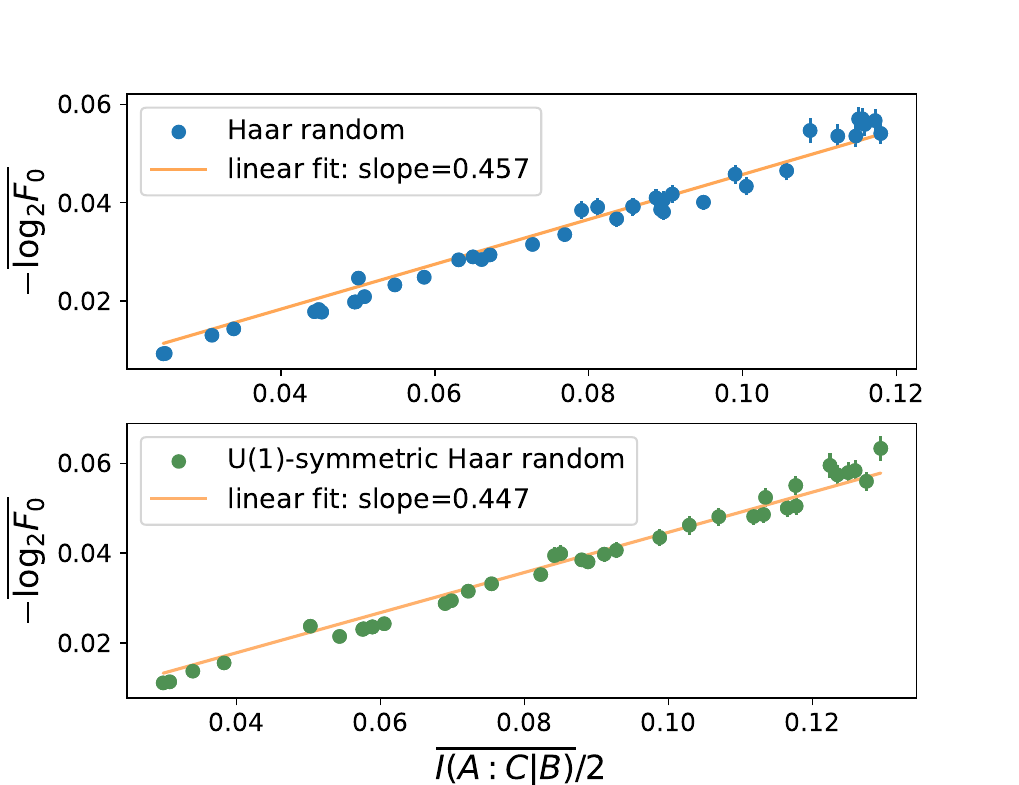}
    \caption{The relation between the averaged Petz map infidelity and the averaged CMI for steady states of MIPTs at criticality with Haar random unitaries (upper) and U(1)-conserving Haar random unitaries (bottom). For both cases, the system size is $L=20$ and we average the results of $N=400$ simulations.}
    \label{fig:F0_CMI_mipt}
\end{figure}

\subsection{Interpretation from boundary conformal field theory}
In this subsection, we show that the linear relation $-\overline{\log F_0} \propto \eta$ generically holds for MIPTs using arguments from conformal field theory. Conformal invariance of MIPTs at criticality has been confirmed extensively in the literature \cite{Bao_2020,Li_2018,Li_2019,Zabalo:2019sfl,Li_2021v2}. Furthermore, the Renyi entropy of a region $A$ is mapped to a two-point correlation function of boundary operators in a logarithmic conformal field theory \cite{Jian:2019mny,Li_2021v2}. We will generalize this argument to the Petz map fidelity and show that $-\overline{\log F_0} \propto \eta$ at small $\eta\ll 1$.

Let us first recall the mapping from the Haar-random MIPT to the statistical mechanics model \cite{Jian:2019mny}. Denote the circuit in Fig.~\ref{fig:mipt} as $C$, which is a product of all the random unitaries and projectors. Note that $C$ is not a unitary operator due to the presence of the projectors, thus the steady state
$|\psi(T)\rangle = C|\psi(0)\rangle/\sqrt{Z}$, where $Z = \langle \psi(0) |C^{\dagger} C|\psi(0)\rangle$ is the probability of the state in the ensemble. The Renyi entropy in a region $A= [x_1,x_2]$ is defined by
\begin{equation}
\label{eq:Sn}
    \begin{aligned}
        S_n(A) ~&=~ \frac{1}{1-n}  \log \tr(\rho^{n}_A) \\
         ~& =~  \frac{1}{1-n}  \log \tr(\rho^{\otimes n} \tau_{n,A}) \\
         ~& =~  \frac{1}{1-n} \lim_{m\rightarrow 0}  \frac{1}{m} \tr(\rho^{\otimes nm} \tau^{\otimes m}_{n,A})~,
    \end{aligned}
\end{equation}
where $\tau_n = (123\cdots n) \in S_n$ is the twist operator for cyclic permutation.
The averaged Renyi entropy $\overline{S}_{n}(A)$ can be mapped into the free energy cost of a twist defect of a stat-mech model in the replica limit,
\begin{equation}
\label{eq:mapping0}
    \overline{S}_{n}(A) ~=~ \frac{1}{1-n}\lim_{m\rightarrow 0} \frac{1}{m} \log (Z_A/Z_0)
\end{equation}
where
\begin{equation}
    \begin{aligned}
        Z_A ~&=~ \mathbb{E}_{C} \tr((C|\psi(0)\rangle\langle\psi(0)| C^{\dagger})^{\otimes (mn+1)}\tau^{\otimes m}_n)~,\\
        Z_0 ~&=~ \mathbb{E}_{C} \tr((C|\psi(0)\rangle\langle\psi(0)| C^{\dagger})^{\otimes (mn+1)}) ~.
    \end{aligned}
\end{equation}
The partition function $Z_0$ describes a 1+1-dimensional stat-mech model on the lower-half plane. Despite the technical derivation using Weingarten arithmetics, we simply note a few facts useful for later. Firstly, the model consists of spins $g_i$ taking value in the permutation group $S_{Q}$, where $Q=mn+1$. Secondly, the partition function $Z_A$ corresponds to imposing the fixed boundary condition $g_A = \tau^{\otimes m}_n$ on the boundary region $A$. Thirdly, the model has a $S_Q\times S_Q$ symmetry that acts as left and right multiplications on the spin. The symmetry is spontaneously broken in the volume law phase and preserved in the area law phase. Finally, the phase transition corresponds to the critical point which possesses conformal symmetry. At the critical point, the ratio of partition functions $Z_A/Z_0$ can be expressed in terms of the correlation function of boundary condition changing operators
\begin{equation}
    \frac{Z_A}{Z_0} ~=~ \langle \Sigma_{(\tau^{\otimes m}_n)^{-1}}(x_1) \Sigma_{\tau^{\otimes m}_n}(x_2)\rangle~,
\end{equation}
where $\Sigma_{g}$ is the boundary condition changing operator that changes from the fixed boundary condition $g$ to the fixed boundary condition $\mathbb{1}$. It is then clear from Eq.~\eqref{eq:mapping0} that all the Renyi entropy has logarithmic scaling at the critical point and the coefficient is given by the scaling dimension of $ \Sigma_{\tau^{\otimes m}_n}$ in the replica limit. For a comprehensive analysis of these boundary operators, see Ref.~\cite{Li_2021v2}.

Now we consider the averaged Petz map infidelity of steady states. We have the following replica expression for the Petz map fidelity \cite{Vardhan:2023pnm}
\begin{equation}
    F^{(n_1,n_2,k)}_0 ~=~ \tr\left(\left(\rho^{n_1}_{BC} \rho^{n_2}_{B} \rho_{AB} \rho^{n_2}_{B} \rho^{n_1}_{BC} \rho_{ABC}\right)^k\right)
\end{equation}
where the replica limit is defined by $n_1\rightarrow 1/2, n_2\rightarrow -1/2, k\rightarrow 1/2$. In this limit, $F^{(n_1,n_2,k)}_0\rightarrow F_0$. In terms of twist operators, the replicated fidelity can be written as
\begin{equation}
    F^{(n_1,n_2,k)}_0 ~=~ \tr (\rho^{\otimes M} \tau_A \tau_B \tau_C)~,
\end{equation}
where $M = 2k(1+n_1 + n_2)$ is the total number of replicas, and $\tau_A,\tau_B,\tau_C \in S_{M}$ specifies the twist operator whose exact form is not important later on and can be found in Ref.~\cite{Vardhan:2023pnm}. Thus, the averaged Petz map infidelity acquires a replica expression which is similar to Eq.~\eqref{eq:Sn},
\begin{equation}
    \log F_0 ~= \lim_{
    \substack{
        m\rightarrow 0, k\rightarrow \frac{1}{2}\\
        n_1 \rightarrow \frac{1}{2}, n_2\rightarrow -\frac{1}{2}
    }} \frac{1}{m}\tr (\rho^{\otimes Mm} \tau^{\otimes m}_A \tau^{\otimes m}_B \tau^{\otimes m}_C)\,.
\end{equation}
Thus, following the same derivation that leads to Eq.~\eqref{eq:mapping0}, we obtain
\begin{equation}
\label{eq:fidelity_4pt}
    \overline{\log F_0} ~=~ \lim_{
    \substack{
        m\rightarrow 0, k\rightarrow \frac{1}{2}\\
        n_1 \rightarrow \frac{1}{2}, n_2\rightarrow -\frac{1}{2}
    }} \frac{1}{m}\log (Z_{\tau_A,\tau_B,\tau_C}/Z_0)~,
\end{equation}
where 
\begin{equation}
    \begin{aligned}
        Z_{\tau_A,\tau_B,\tau_C} ~&=~ \mathbb{E}_{C} \tr\Big((C|\psi(0)\rangle\langle\psi(0)| C^{\dagger})^{\otimes (mM+1)} \\
        &\qquad\qquad\qquad\qquad\qquad  \tau^{\otimes m}_A \tau^{\otimes m}_B \tau^{\otimes m}_C\Big) ~,\\
        Z_0 ~&=~ \mathbb{E}_{C} \tr((C|\psi(0)\rangle\langle\psi(0)| C^{\dagger})^{\otimes (mM+1)})~.
    \end{aligned}
\end{equation}
The upshot is that the bulk partition function $Z_0$ is the same as the $Z_0$ for $M$-th Renyi entropy in Eq.~\eqref{eq:mapping0}, but the boundary condition in $Z_{\tau_A,\tau_B,\tau_C}$ is different: on each interval $A = [x_1,x_2], B=[x_2,x_3],C=[x_3,x_4]$, a fixed boundary condition is imposed. As a result, the ratio of partition functions is a four-point correlation function of the boundary operators,
\begin{equation}
    \frac{Z_{\tau_A,\tau_B,\tau_C}}{Z_0} ~=~ \langle \Sigma_1(x_1) \Sigma_2(x_2) \Sigma_3(x_3) \Sigma_4(x_4) \rangle~,
\end{equation}
where 
\begin{equation}
    \begin{split}
        \Sigma_1 ~=&~ \Sigma_{(\tau^{\otimes m}_A)^{-1}}~,~ \Sigma_2 ~=~ \Sigma_{(\tau^{-1}_B \tau_A)^{\otimes m}}~,\\
        \Sigma_3 ~=&~ \Sigma_{(\tau^{-1}_C \tau_B)^{\otimes m}}~,~ \Sigma_4 ~=~  \Sigma_{\tau^{\otimes m}_C}~.
    \end{split}
\end{equation}
In the replica limit, all the scaling dimensions go to zero and the four-point correlation function only depends on the cross-ratio $\eta$ due to conformal invariance
\begin{equation}
\label{eq:4pt}
    \frac{Z_{\tau_A,\tau_B,\tau_C}}{Z_0} ~=~ \langle \Sigma_1(0) \Sigma_2(\eta) \Sigma_3(1) \Sigma_4(\infty)\rangle~.
\end{equation}
As $\eta\rightarrow 0$, we can use the operator product expansion to expand the four-point correlation function to first order in $\eta$ \cite{Marshakov_2010}
\begin{equation}\resizebox{0.41\textwidth}{!}{$%
    \scalemath{1.0}{\frac{Z_{\tau_A,\tau_B,\tau_C}}{Z_0}} = C_{12} C_{34} \left[1+\eta\scalemath{1.0}{\frac{(h+h_1-h_2)(h+h_3-h_4)}{2h}}\right], $}%
\end{equation}
where $h_i$ is the scaling dimension of $\Sigma_i$ and $h$ is the scaling dimension of $\Sigma = \Sigma_1\times \Sigma_2 = \Sigma_{(\tau^{-1}_B)^{\otimes m}}$. Note that all scaling dimensions are proportional to $m$ as $m\rightarrow 0$, thus we can define $\tilde{h}_i := \partial h_i/\partial m|_{m=0}$ in the replica limit.

As $\eta\rightarrow 0$, we have $\overline{\log F_0} = 0$, thus $C_{12}C_{34} = 1+O(m^2)$. Thus, Eq.~\eqref{eq:fidelity_4pt} implies that
\begin{equation}
    \overline{\log F_0}  ~=~ \frac{(\tilde{h}+\tilde{h}_1-\tilde{h}_2)(\tilde{h}+\tilde{h}_3-\tilde{h}_4)}{2\tilde{h}} \eta + O(\eta^2)~.
\end{equation}
This explains the linear relation observed by the numerics.

Several comments are in order. First, different MIPTs correspond to different stat-mech models. However, the derivation above essentially only uses conformal invariance and OPEs. Thus, the linear relation persists for all MIPTs. Second, the linear coefficient depends on the scaling dimensions of particular boundary condition changing operators in the replica limit. Different MIPTs have different operator contents, thus, the linear coefficient varies across all MIPTs. Third, the derivation also suggests that $-\overline{\log F_t} \propto \eta$ for all $t$, albeit with a $t$-dependent linear coefficient. We leave it to future work for a systematic analysis of the scaling dimensions and OPEs for different replicas and the replica limits. Finally, our numerics also suggest that the linear relation holds for the volume law phase of MIPTs, which cannot be explained by the argument from conformal symmetry above. It is known that the entanglement of MIPTs in the volume law phase is given by the KPZ universality class \cite{Li_KPZ2023}. Thus, one may explain the linear relation by mapping the Petz map fidelity to observables in the KPZ dynamics, which we leave for future work.

\section{Petz map for quantum critical states}
\label{sec:CFT}
In this section, we consider one-dimensional quantum critical states that are described by conformal field theory (CFT) at long distances. Quantum critical states can be realized in the following two ways. Firstly, they are ground states of the local Hamiltonian tuned to the critical gapless point. Secondly, they appear on the boundary of a higher-dimensional topological phase or symmetry-protected topological phase. In the second scenario, the critical state can be chiral if the bulk phase has non-zero thermal Hall conductance, a phenomenon known as bulk-boundary correspondence. In the following, we compute the Petz map fidelity for chiral and non-chiral critical states. We will see that the non-chiral and chiral critical states differ from each other and that both differ from the MIPT steady states in the previous section. In addition, we consider the ensemble obtained by performing local measurements on quantum critical states. We find that local measurement does not alter the behavior of the Petz map fidelity. Our calculation of Petz map fidelity is based on numerical simulations of fermionic Gaussian states \cite{Swingle:2018dto}, which allow us to go to large system sizes up to $L=128$. Many details of the numerical techniques are reviewed in Appendix.~\ref{appen:fermion}.  
\subsection{Non-chiral states}
The ground state of one-dimensional gapless local Hamiltonians can be usually described by non-chiral CFT at long distances. Let us for concreteness consider a free fermion chain of $2L$ Majorana fermions,
\begin{equation}
\label{eq:FF}
    H ~=~ i \sum_{k=1}^{2L} \gamma_k \gamma_{k+1}~,
\end{equation}
where the canonical anticommutation relation $\{\gamma_k,\gamma_l\} = 2\delta_{kl}$ is satisfied. We impose the Neveu-Schwarz boundary condition $\gamma_{2L+1} = -\gamma_{1}$ to ensure a unique ground state. The low-energy physics is described by a free fermion CFT with central charge $c=1/2$. The entanglement entropy of an interval is given by \cite{Calabrese_2004}
\begin{equation}
\label{eq:S_CFT}
    S(A) ~=~ \frac{c}{3}\log L_A + O(1)~,
\end{equation}
which is of the same form as the MIPT at criticality. Thus, the CMI also takes the form of Eq.~\eqref{eq:CMI_CFT} with $\alpha = c/3$.  
We consider the Petz map recovery for contiguous intervals $A,B$ and $C$ on this state. The ground state is a fermionic Gaussian state (see Appendix~\ref{app:state} for general definitions and properties). We are able to simulate large system sizes up to $L=128$ with the help of Gaussian fermion techniques introduced in Ref.~\cite{Swingle:2018dto}. We note in passing that the bosonic analogue of the technique is developed in Ref.~\cite{Lami_2018}. We verify that while the CMI is proportional to $\eta$ at small $\eta$ (see Appendix~\ref{app:entropy} for details), and that the best Petz map infidelity $-\log F_0 \propto \eta^2$ (see Appendix~\ref{app:petz} for details), consistent with Ref.~\cite{Vardhan:2023pnm}. The fidelity $F_t$ of the rotated Petz map is symmetric with respect to $t$, signaling time-reversal invariance, see Fig.~\ref{fig:Ising_plot}.
\begin{figure}[htp!]
    \centering
    \includegraphics[width=0.5\textwidth]{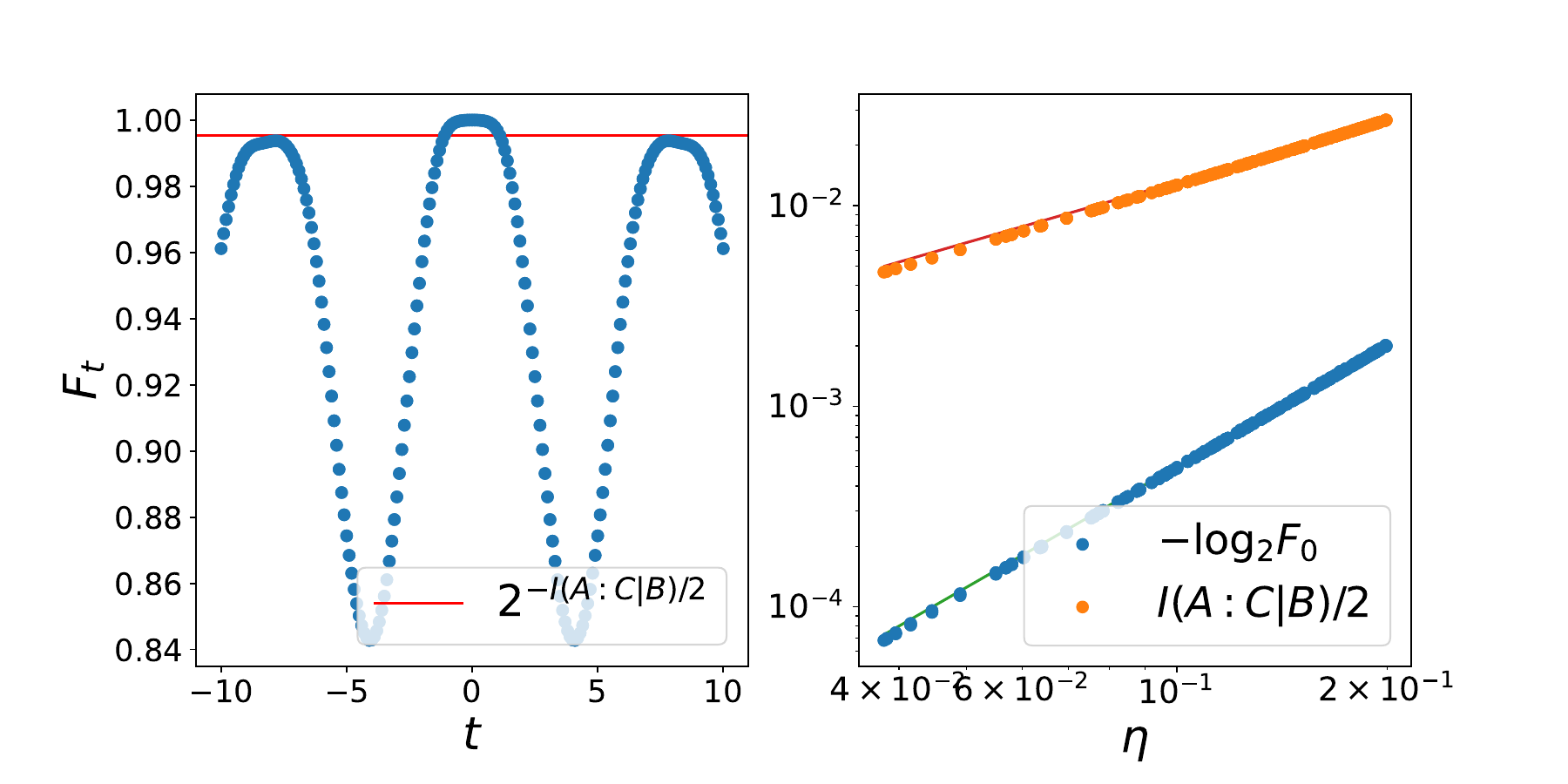}
    \caption{Left: the relation between the rotated Petz map fidelity $F_t$ and parameter $t$ under a fixed subsystem configuration ($\eta = 0.055$) for the ground state of Eq.~(\ref{eq:FF}) with $L=128$. 
    %Here we take $L=128$, $L_A = L_C = 8$, $L_B = 32$ ($\eta = 0.055$). 
    Right: behavior of $-\log F_0$ and the CMI for the same state. The linear fit (in red) and the quadratic fit (in green) show the dependence of the form $I(A:C|B) \propto \eta$ and $-\log F_0 \propto \eta^2$.}
    \label{fig:Ising_plot}
\end{figure}

Now we show that the above behavior persists for an ensemble of critical states under local measurements. Measurements of observables quadratic in fermion operators are Gaussian maps (see Appendix~\ref{app:map} for a general discussion of Gaussian maps), i.e., they map fermionic Gaussian states to an ensemble of fermionic Gaussian states labelled by different measurement outcomes. It has been shown in Ref.~\cite{Weinstein_2023} that the ensemble of states obtained by measuring fermion parity $-i\gamma_{2j-1}\gamma_{2j}$ on each site with probability $p<1$ have the same scaling of entropy Eq.~\eqref{eq:S_CFT} upon averaging over trajectories. It follows that the averaged CMI has the same linear scaling at small $\eta$ as in Eq.~\eqref{eq:CMI_smalleta}. We consider the Petz map on the measurement ensemble, see Appendix~\ref{appen:measurement} for details of numerical implementation. As shown in Fig.~\ref{fig:Ising_measurement_plot},  the average infidelity $-\overline{\log F} \propto \eta^2$, similar to the ground state. Although the measurement ensemble of CFT ground states may have the same logarithmic scaling of average entanglement entropy as the MIPTs at criticality, the two ensembles have very different multipartite entanglement among more than two regions, as is evident by the scaling of averaged Petz map infidelity.

\begin{figure}[htp!]
    \centering
    \includegraphics[width=0.5\textwidth]{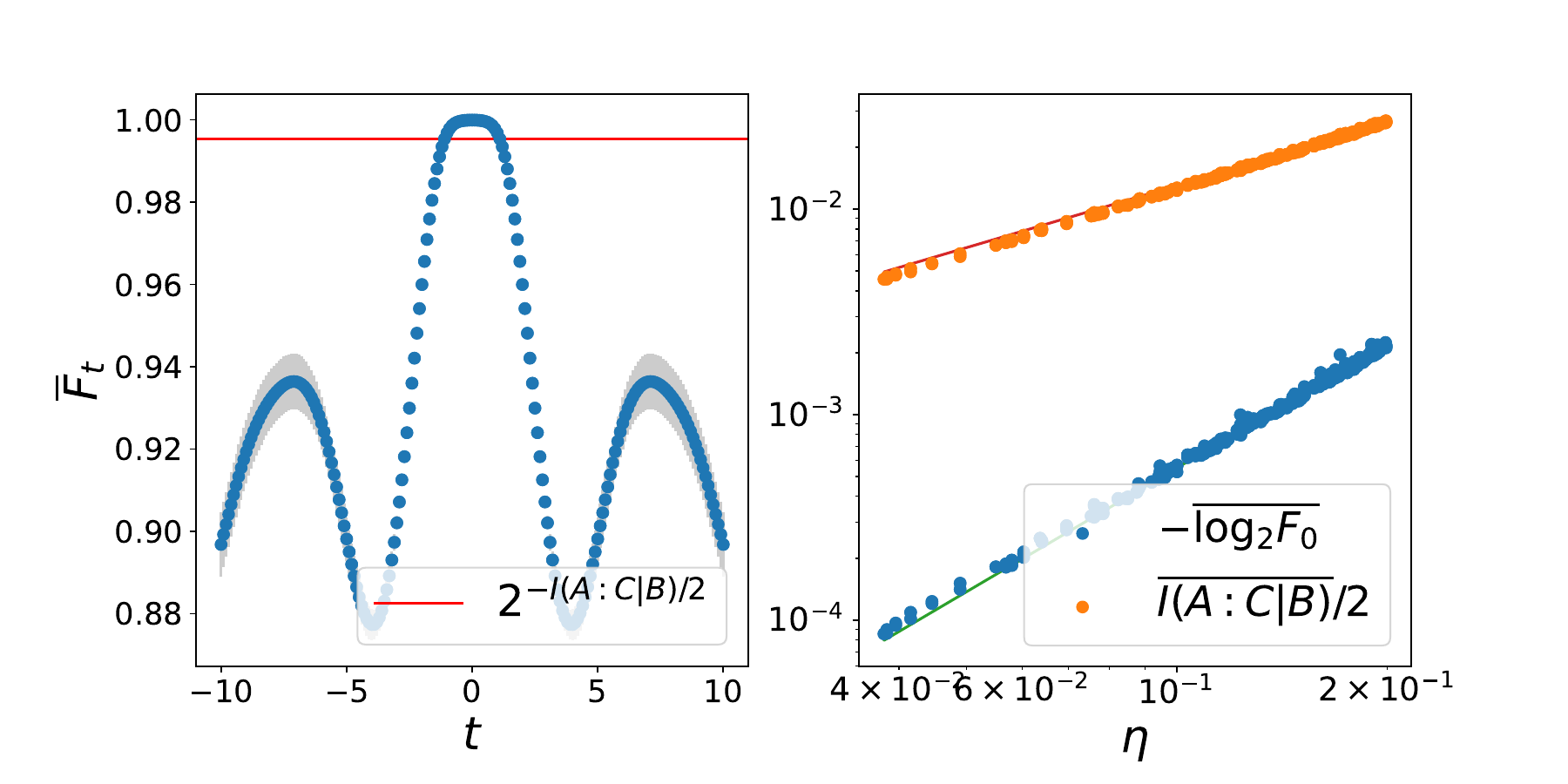}
    \caption{Left: the relation between the averaged $F_t$ and $t$ for measurement ensemble. We take the measurement rate as $p=0.2$ and average the results of $N=100$ simulations. Error bars are shown in gray. 
    Right: behaviors of the averaged infidelity and the CMI. The linear fit is in red and the quadratic fit is in green.}
    \label{fig:Ising_measurement_plot}
\end{figure}

\subsection{Chiral states}
The rotated Petz map Eq.~\eqref{equ:rotated_Petz} has a rotation parameter $t$ which changes sign under time-reversal symmetry. The asymmetry of the Petz map fidelity $F_t$ with respect to $t\rightarrow -t$ then naturally indicates breaking of time reversal symmetry. Here we illustrate the asymmetry for the chiral free fermion CFT state, which appears on the boundary of a 2+1 dimensional $p+ip$ topological superconductor. The chirality can also be indicated by other multipartite entanglement quantities, such as modular commutators, see Refs.~\cite{Kim_2022,Fan_2022,Zou_2022v2,gass2024manybodysystemsspuriousmodular}.

Recall that the reduced density matrix of the boundary of a chiral topological order can be described by a chiral thermal state with different temperatures $\beta_L\neq \beta_R$ for left and right moving modes \cite{Qi_2012},
\begin{equation}
    \rho_{\beta_L,\beta_R} ~\propto~ e^{-\beta_L H_L -\beta_R H_R}~,
\end{equation}
where $H_L$ and $H_R$ are Hamiltonians of chiral CFTs with only left or right moving modes. If the boundary has a left-moving mode in its ground state, then one can show that $\beta_L = \infty$ and $\beta_R=O(1)$. For the free fermion model Eq.~\eqref{eq:FF}, one can diagonalize the Hamiltonian in the momentum space and construct $H_L$ and $H_R$ by selecting left or right moving momenta, respectively. Thus, we can construct the correlation matrix of $\rho_{\beta_L,\beta_R}$ using Gaussian fermion techniques (see appendix~\ref{appen:CFT} for details). The CMI is half of that of the non-chiral ground state, because only left-moving modes contribute. Computing fidelity of the rotated Petz map, we find that the fidelity is almost constant for $t>0$ and drops rapidly as $t<0$, see Fig.~\ref{fig:chiral_Ising_plot}. The infidelity for $t>0$ again scales as $\eta^2$, signaling the underlying conformal field theory. We also consider the ensemble of mixed states obtained by measuring local fermion parity. The result is qualitatively similar to the chiral thermal state with no measurement, see Fig.~\ref{fig:chiral_Ising_measurement_plot}. This is an indication that the chiral fermion CFT on the boundary of $p+ip$ superconductor is robust under local parity measurements.

\begin{figure}[htp!]
    \centering
    \includegraphics[width=0.5\textwidth]{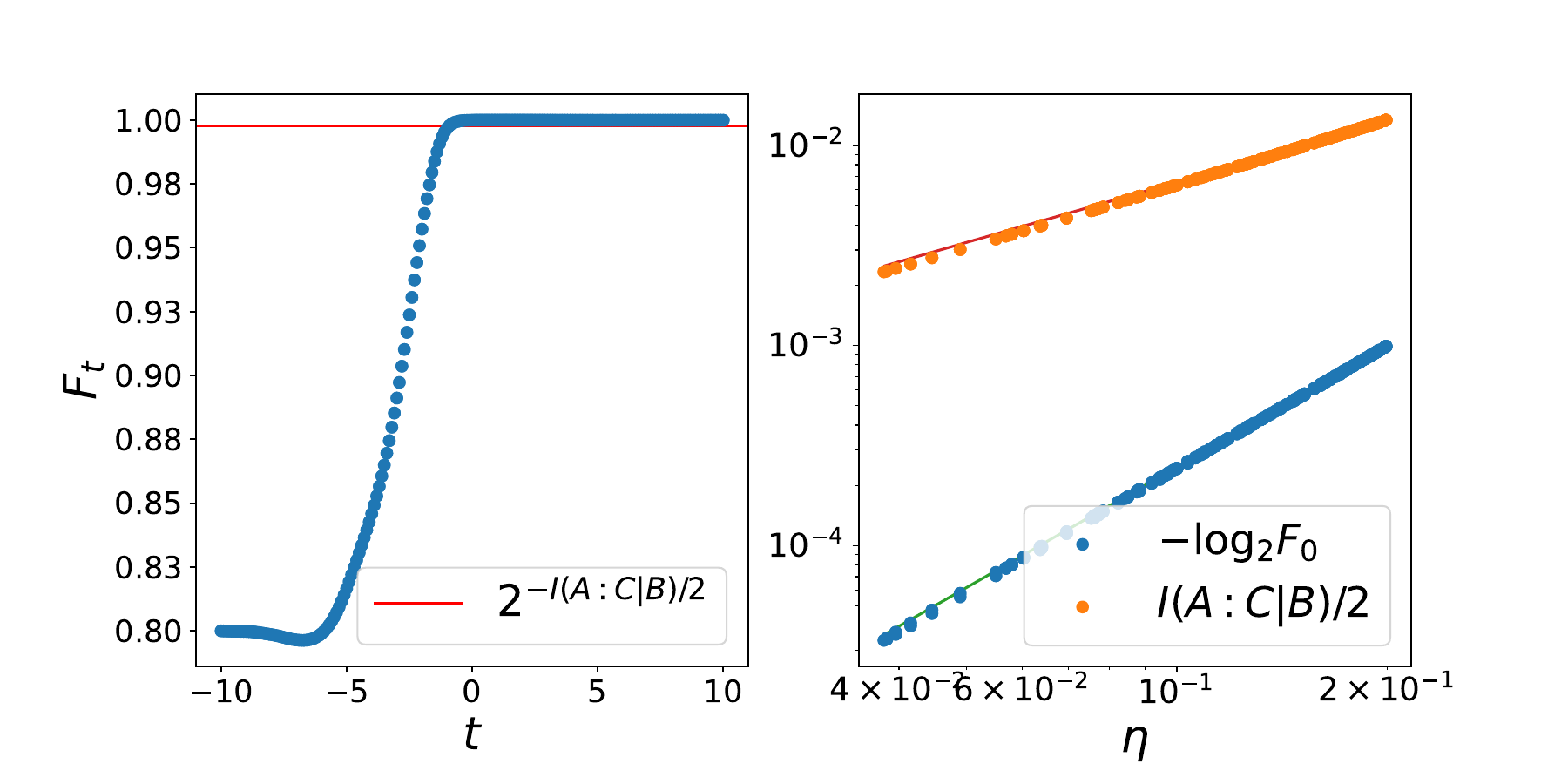}
    \caption{Left: the relation between $F_t$ and $t$ for the chiral state ($\beta_R=8)$. 
    Right: behavior of the infidelity and the CMI. The linear fit is in red and the quadratic fit is in green.}
    \label{fig:chiral_Ising_plot}
\end{figure}

\begin{figure}[htp!]
    \centering
    \includegraphics[width=0.5\textwidth]{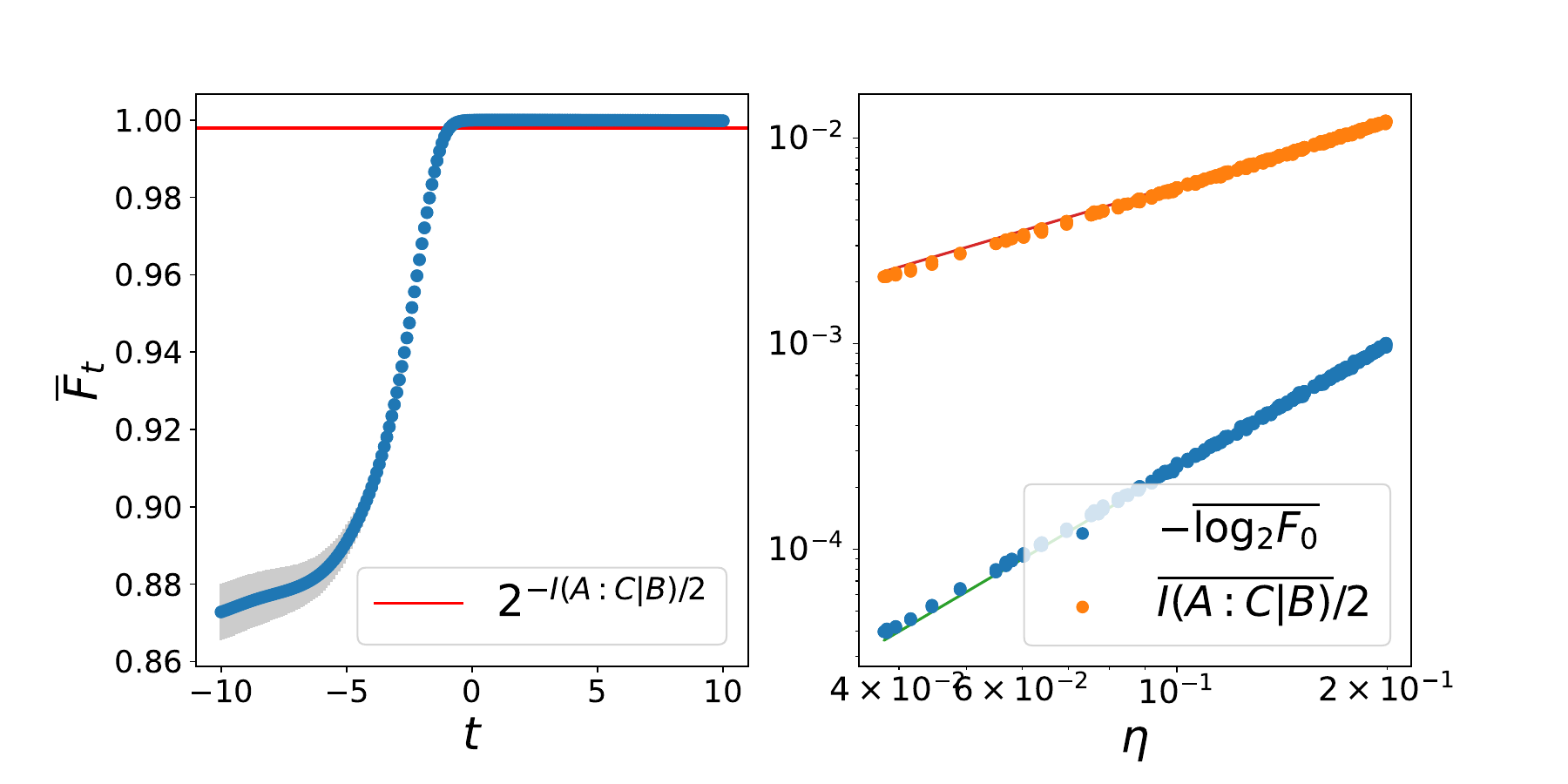}
    \caption{Left: the relation between the averaged fidelity and $t$ for the chiral state with measurements. 
    Right: behavior of $-\overline{\log F_0}$ and the averaged CMI. The linear fit is in red and the quadratic fit is in green.}
    \label{fig:chiral_Ising_measurement_plot}
\end{figure}

\section{Petz map recovery for topological order}
\label{sec:TO}
In this section, we show that the Petz map offers an operational interpretation of the topological entanglement entropy (TEE)~\cite{Levin_2006,Kitaev_2006} in a topologically ordered system. Our discussion is a slight generalization of those of Ref.~\cite{Shi_2020}.

We consider the Levin-Wen partition \cite{Levin_2006} of topological order shown in Fig.~\ref{fig:TEE}, where the TEE formula implies $I(A:C|B) = 2\log D$, where $D = \sqrt{\sum_{a}d^2_a}$ is the total anyon dimension. While useful for most cases in practice, there are still spurious cases \cite{Williamson_2019} where the TEE formula fails. It has been recently proven \cite{Kim_2023} that finite-depth local unitary circuits can only increase the TEE, thus $I(A:C|B)\geq 2\log D$. 

In the following, we derive $I(A:C|B)\geq 2\log D$ from an operational perspective. We focus on the case where all anyons are Abelian, where $d_a = 1$ and $D$ is the total number of anyons. We make two physically-motivated assumptions. (I) We assume that by inserting an anyon inside the white triangle, the reduced density matrices $\rho^{(a)}_{ABC}$ on $ABC$ are perfectly distinguishable and have the same CMI. (II) We assume that $\rho^{(a)}_{AB}$ and $\rho^{(a)}_{BC}$ are independent of the anyon type $a$. The physical reason underlying these assumptions is that anyons can be created by inserting a fattened Wilson line operator, which could bypass any single region $A$, $B$, or $C$. These assumptions can be proven with a stronger set of entanglement bootstrap axioms \cite{Shi_2020}, but here we will not assume those axioms. With these two assumptions, we can prove the following theorem based on the twirled Petz map.

\begin{figure}
    \centering
    \includegraphics[width = 0.6\linewidth]{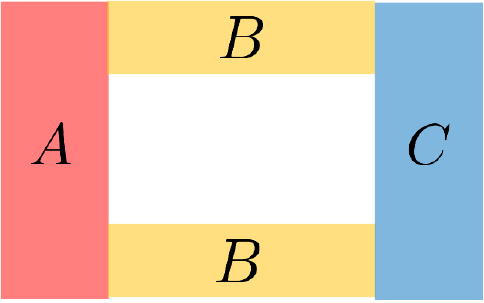}
    \caption{The Levin-Wen partition of topological entanglement entropy}
    \label{fig:TEE}
\end{figure}
\begin{thm}
For the Levin-Wen partition of the bulk of an Abelian topological order, $I(A:C|B)\geq 2\log D$.
\end{thm}
\begin{proof}
    Consider the twirled Petz map \eqref{eq:twirled} for the erasure channel on $C$ and a single anyon $a$ inserted in the bulk. The recovered state $\tilde{\rho}^{(a)}_{ABC}$ is independent of the anyon type $a$ because of our assumption (II). We will henceforth denote this state as $\tilde{\rho}_{ABC}$. Eq.~\eqref{eq:fid_bound_1} then implies that $F(\tilde{\rho}_{ABC}, \rho^{(a)}_{ABC}) \geq 2^{-I(A:C|B)/2}$. Suppose $I(A:C|B)<2\log D$, then $F(\tilde{\rho}_{ABC}, \rho^{(a)}_{ABC}) > 1/D$. Now we use our assumption (I) to deduce a contradiction. By Uhlmann's theorem, there exists a purification $|\tilde{\psi}\rangle$ of $\tilde{\rho}_{ABC}$ and a series of purification $|\psi^{(a)}\rangle$ of $\rho^{(a)}_{ABC}$ which satisfy $|\langle\tilde{\psi}|\psi^{(a)}\rangle| = F(\tilde{\rho}_{ABC}, \rho^{(a)}_{ABC})>1/D$. Since $\rho^{(a)}_{ABC}$'s are perfectly distinguishable, all their purification $|\psi^{(a)}\rangle$'s are also perfectly distinguishable, that is, $\langle \psi^{(a)}| \psi^{(b)} \rangle = \delta^{ab}$. The contradiction follows from $1 = \langle \tilde{\psi}|\tilde{\psi}\rangle \geq \sum_{a} |\langle \tilde{\psi}|\psi^{(a)}\rangle|^2 >1$. Thus we conclude $I(A:C|B)\geq 2\log D$.
\end{proof}
The underlying intuition is that one cannot recover the information of anyon insertion from any two of the subsystems. There is possibly more information that one cannot recover, which corresponds to spurious TEE. The Petz map, however, cannot in general tell the difference between the anyon contribution and the spurious contribution. As we have shown in Sec.~\ref{sec:MIPT}, for stabilizer states the Petz map fidelity is determined by the CMI. It is open whether there may exist non-stabilizer examples where the Petz map can tell us more about spurious TEE than the CMI alone.

\section{Conclusion}
\label{sec:conc}
In this work, we have shown that the Petz map fidelity is a useful diagnostic of quantum phases, which may offer information beyond the CMI. We show that for the steady states of MIPTs at criticality, the Petz map infidelity exhibits a linear relation with the CMI at small cross-ratios. This is to be contrasted with critical ground states which have a quadratic relation between the infidelity and the CMI. The quadratic relation persists under local measurements, retaining the logarithmic scaling of average entanglement entropy. Furthermore, we consider the boundary of chiral topological order and show that the Petz map fidelity is asymmetric with respect to the rotation parameter. Such an asymmetry also persists when adding local measurements, signaling chirality for the post-measurement ensemble. For topological order, the Petz map offers an operational way to interpret topological entanglement entropy. The interpretation also implies nonnegative spurious contributions under some physical assumptions.

We anticipate that the Petz map may offer entanglement diagnostic beyond the CMI in various other settings, such as recently discovered decoherence-induced phase transitions in topological order and quantum critical states. From an information-theoretic perspective, the quantitative features of the Petz map may be used to distinguish different classes of sequences of approximate quantum Markov chains. The rotated Petz map is particularly useful to indicate the breaking of time-reversal symmetry. We leave it to future work for a systematic discussion of chirality in quantum many-body states.

\section*{Acknowledgements}
Y.Z. is deeply grateful to Yaodong Li and Shengqi Sang for very patiently and insightfully teaching us their works on measurement-induced phase transitions. Y.Z. acknowledges Shreya Vardhan, Bowen Shi, and Isaac Kim for collaboration on a related work. Y.H. thanks William Witczak-Krempa for valuable discussions. We thank Mark Wilde for helpful comments.
The research of Y.H. is supported by the Celestial Holography Initiative at the Perimeter Institute for Theoretical Physics and the Simons Collaboration on Celestial Holography. 
Research at the Perimeter Institute is supported by the Government of Canada through the Department of Innovation, Science and Industry Canada and by the Province of Ontario through the Ministry of Colleges and Universities.

\appendix
\section{Fermionic Gaussian States and Operations}\label{appen:fermion}
One of the challenges in classical simulations for quantum many-body systems is the exponential growth of the resources needed as the system size $L$ increases. However, when the Hamiltonian can be expressed in the form of free fermions, the Grassmann representation formalism allows us to reduce the computational cost to $O(L^3)$~\cite{bravyi2004lagrangian} when simulating Gaussian states and Gaussian operations. 
In section~\ref{sec:CFT}, we use this formalism to compute quantities of interest for critical Ising ground states and chiral states with(out) measurements. This appendix is devoted to a self-contained review of fermionic Gaussian states and Gaussian linear maps, followed by explicit applications relevant to the results presented in section~\ref{sec:CFT}. The material discussed here mainly follows the results from~\cite{bravyi2004lagrangian,Swingle:2018dto}.
We have been developing a Python package $\mathfrak{QM}$atch utilizing this formalism to simulate fermionic Gaussian states and operations. The package can be accessed \href{https://github.com/yhu1996/QMatch}{here} on GitHub. 

\subsection{Gaussian States}
\label{app:state}
We consider a system composed of $L$ fermions which can be described by fermionic creation ($c^{\dagger}_i$) and annihilation ($c_i$) operators. They satisfy the anti-commutation relation $\{ c_i,c^{\dagger}_j\} = \delta_{ij}\,\mathbb{I}$. This system can also be described by a set of $2L$ Majorana fermions $\gamma_i$, which satisfy the Clifford algebra $\{\gamma_i,\gamma_j\}=2\delta_{ij}\mathbb{I}$. 
The transformation between these two equivalent descriptions is given by 
\begin{equation}
    c_j^{\dagger} ~=~ \frac{1}{2}(\gamma_{2j-1} + i\,\gamma_{2j}) ~,~ c_j ~=~ \frac{1}{2}(\gamma_{2j-1} - i\,\gamma_{2j}) ~.
    \label{equ:c-gamma}
\end{equation}
In what follows, we stick with the Majorana fermion formalism.

Note that the operator algebra of $L$ fermionic modes is essentially the set of operators that can be represented as a polynomial of $\gamma_i$'s, where $i=1,\dots,2L$. Let's call it ${\cal C}_{2L}$. For any operator $X\in {\cal C}_{2L}$, it has a Grassmann representation $\omega(X,\theta)$ where $\theta$'s are Grassmann variables ($\theta_i\theta_j = - \theta_j\theta_i$ and $\theta_i^2=0$) and 
\begin{equation}
    \omega(\gamma_i\gamma_j\cdots\gamma_l,\theta) ~=~ \theta_i\theta_j\cdots\theta_l ~~,~~ \omega(\mathbb{I},\theta) ~=~ 1 ~~.
\end{equation}

A fermionic quantum state is Gaussian if and only if its density matrix $\rho$ has a Gaussian Grassmann representation, namely
\begin{equation}
    \omega(\rho,\theta) ~=~ \frac{1}{2^L}\,\exp\left( \frac{i}{2}\,\theta^{\rm T}\,G^{\rho}\,\theta \right) ~~,
    \label{equ:gaussian_state}
\end{equation}
where $\theta^{\rm T}G^{\rho}\,\theta := \sum_{i,j}(G^{\rho})_{ij}\theta_i\theta_j$.
Due to the anti-commuting property of $\theta$, $G^{\rho}$ is an antisymmetric real $2L\times 2L$ matrix and thus admits the Schur decomposition in terms of an orthogonal matrix $O\in {\rm O}(2L)$  
and a block diagonal skew-symmetric matrix $\Lambda$ as follows
\begin{equation}
    G^{\rho} ~=~ O\,\Lambda\,O^{\rm T} ~,~ \Lambda = {\rm diag}(\epsilon_1,\cdots,\epsilon_L)\otimes
    \begin{pmatrix}
        0 & 1\\
        -1 & 0
    \end{pmatrix} ~.
    \label{equ:antisym_decomp}
\end{equation}

The matrix $G^{\rho}$ is called the correlation matrix of $\rho$. By using the definition Eq.~(\ref{equ:gaussian_state}) one can check that 
\begin{equation}
    (G^{\rho})_{ij} ~=~ \begin{cases}
        0 \qquad\qquad\qquad\qquad i=j \\
        {\rm Tr}\Big(\,\rho\,i\gamma_i\gamma_j \Big) \qquad\quad~ i\ne j
    \end{cases} ~~.
    \label{equ:Cij-def}
\end{equation}
The antisymmeric nature of $G^{\rho}$ is manifest in Eq.~(\ref{equ:Cij-def}). Also note that if $\rho\propto\mathbb{I}$ is a maximal mixed state, $G^{\rho}$ is a zero matrix.

The fermionic Gaussian states can also be viewed as thermal states $e^{-\beta H}$ (as well as ground states as $\beta\rightarrow \infty$) of Hamiltonian quadratic in Majorana fermions. The generic form of such a Hamiltonian reads 
\begin{equation}
    H ~=~ \frac{i}{2}\,\sum_{i,j=1}^{2L}\,M_{ij}\,\gamma_i\,\gamma_j ~~,
    \label{equ:H}
\end{equation}
where $M$ is also an antisymmetric real $2L\times 2L$ matrix. Below, in appendix~\ref{appen:CFT}, we will show explicitly how to compute the correlation matrix for the ground state or thermal states of $H$.

\subsection{Gaussian Linear Maps}
\label{app:map}
A linear map ${\cal E}:\,{\cal C}_{2L}\to {\cal C}_{2L}$ is Gaussian if and only if it has the following Grassmann representation 
\begin{equation} 
    \begin{split}
        & \omega\left( {\cal E}(X), \theta\right) \\
        ~=&~ c\,\int D\mu D\eta\, \exp\Big[ S(\theta,\eta) \,+\, i\eta^{\rm T}\mu \Big]\omega(X,\mu)~,
    \end{split}
    \label{equ:Gaussian_map}
\end{equation}
where $X\in {\cal C}_{2L}$ is the input, $\theta$, $\mu$, and $\eta$ are all Grassmann variables, and $\int D\mu:=\prod_{i=1}^L\int d\mu_i$. Properties of integrals over Grassmann variables can be found in many quantum field theory textbooks. Readers may also refer to~\cite{bravyi2004lagrangian}. 
Moreover, $c$ is a complex number, $\eta^{\rm T}\mu:=\sum_i\eta_i\mu_i$, and
\begin{equation}
    S(\theta,\eta) ~=~ \frac{i}{2}\,\theta^{\rm T}\mA\,\theta ~+~ \frac{i}{2}\, \eta^{\rm T}\mD\,\eta ~+~ i\,\theta^{\rm T}\mB\,\eta ~~.
    \label{equ:S(theta,eta)}
\end{equation}
Note that $\mA$, $\mB$, and $\mD$ are $2L\times 2L$ matrices and $\mA$ and $\mD$ are antisymmetric. 
By plugging Eq.~(\ref{equ:gaussian_state}) into Eq.~(\ref{equ:Gaussian_map}), we can see the new correlation matrix can be written in terms of the original one $G^{X}$ as follows
\begin{equation}
    G^{{\cal E}(X)} ~=~ \mA ~+~ \mB\,\left( \mD + \left(G^{X}\right)^{-1} \right)^{-1} \,\mB^{\rm T} ~~,
    \label{equ:Cij_transf}
\end{equation}
where we have used the following integral twice
\begin{equation}
\begin{split}
    &\int D\theta\,\exp\left( \frac{i}{2}\,\theta^{\rm T}M\,\theta + \eta^{\rm T}\mu \right) \\
    & \qquad\qquad ~=~ i^L\,{\rm Pf}(M)\,\exp\left( -\frac{i}{2}\,\eta^{\rm T}M^{-1}\eta \right) ~~.
\end{split}
\end{equation}
Here ${\rm Pf}(M)$ is the Pfaffian of an anti-symmetric matrix $M$ and contributes to the transformation of the coefficient $c$. Note that $c$ is irrelevant for physics since it is fixed by the normalization $\tr \rho = 1$, we omit the $c$ transformation here and its explicit expression can be found in~\cite{Swingle:2018dto}.

In short, under the integral representation Eq.~(\ref{equ:Gaussian_map}), the Gaussian linear map is completely determined by four matrices/factor $\mA$, $\mB$, $c$, and $\mD$. Note that a Gaussian linear map is not necessarily a quantum channel. If ${\cal E}(X)$ is trace-preserving, it has $c=1$ and $\mD=0$. Thus the transformation of correlation matrices Eq.~(\ref{equ:Cij_transf}) takes a simplified form 
\begin{equation}
    G^{{\cal E}(X)} ~=~ \mA ~+~ \mB\,G^{X} \,\mB^{\rm T} ~~.
    \label{equ:Cij_transf_TP}
\end{equation}
The transformation of the correlation matrix under a fermionic Gaussian channel follows Eq.~(\ref{equ:Cij_transf_TP}).

Finally, for a given fermionic Gaussian channel, its integral representation Eq.~(\ref{equ:Gaussian_map}) can be found systematically using the Choi–Jamiolkowski isomorphism~\cite{bravyi2004lagrangian}. 
The isomorphism says for a Gaussian linear map ${\cal E}:\,{\cal C}_{2L}\to {\cal C}_{2L}$, its dual operator is defined as 
\begin{equation}
    \rho_{\cal E} ~=~ ({\cal E}\otimes_f \mathbb{I})(\rho_I)~\in~{\cal C}_{4L}
\end{equation}
where 
\begin{equation}
    \begin{split}
        & {\cal E}_1\otimes_f{\cal E}_2(\gamma_1,\dots,\gamma_{2L};\gamma_{2L+1},\dots,\gamma_{4L})\\
        ~=&~ {\cal E}_1(\gamma_1,\dots,\gamma_{2L})\,{\cal E}_2(\gamma_{2L+1},\dots,\gamma_{4L})~.
    \end{split}
\end{equation}
$\rho_I$ is a pure Gaussian state and can be viewed as a maximally entangled state between $\gamma_1,\dots,\gamma_{2L}$ and  $\gamma_{2L+1},\dots,\gamma_{4L}$, which reads
\begin{equation}
    \rho_I ~=~ \frac{1}{2^{2L}}\,\prod_{i=1}^{2L}(\mathbb{I}+i\gamma_i\gamma_{2L+i}) ~.
\end{equation}
To obtain the Grassmann representation of $\rho_{\cal E}$, we assign $\gamma_1,\dots,\gamma_{2L}$ $\to$ $\theta_1,\dots,\theta_{2L}$ and $\gamma_{2L+1},\dots,\gamma_{4L}$ $\to$ $\eta_1,\dots,\eta_{2L}$. Namely, 
\begin{equation}
    \omega(\rho_I, \theta,\eta) ~=~ \frac{1}{2^{2L}}\,\exp\Big( i\,\theta^T\,\eta\Big)~.
\end{equation}
Plugging in Eq.~(\ref{equ:Gaussian_map}) and evaluating Grassmann integrals yields
\begin{equation}
    \omega(\rho_{\cal E}, \theta,\eta) ~=~ \frac{c}{2^{2L}}\,\exp S(\theta,\eta) ~.
    \label{equ:rho_E-rep}
\end{equation}
Thus, the integral representation of the map ${\cal E}$ can be read off from the integral representation of its dual operator. 
We will see an explicit example below in appendix~\ref{appen:measurement}. 

\subsection{Von Neumann Entropy}
\label{app:entropy}
Given the correlation matrix $G^{\rho}$ of the Gaussian state $\rho$, we next derive the expression of the von Neumann entropy $S(\rho)=-\tr(\rho\log\rho)$ in terms of $G^{\rho}$. 

First, recall that $\rho$ can be viewed as a thermal state w.r.t. the modular Hamiltonian $H_E$ (up to normalization), we thus write 
\begin{equation}
    \rho ~=~ e^{-H_E}/{\rm Tr}\left(e^{-H_E}\right)~,
    \label{equ:rho_ansatz}
\end{equation}
where $H_E = \frac{i}{2}\sum_{i,j}h_{ij}\gamma_i\gamma_j$ and the matrix $h$ can be decomposed as follows
\begin{equation}
    \begin{split}
         h ~=~ O_h\left( \oplus_n\begin{pmatrix}
            0 & \epsilon^h_n\\
           -\epsilon^h_n & 0 
        \end{pmatrix} \right)O_h^T
    \end{split} ~.
\end{equation}
Given the ansatz Eq.~(\ref{equ:rho_ansatz}), we have
\begin{equation}
    \begin{split}
        S(\rho)~=&~ -\,\tr(\rho\log\rho)\\
        ~=&~ \tr(\rho H_E) ~+~ \log {\rm Tr}\left(e^{-H_E}\right)\\
        ~=&~ \frac{1}{2}\,\sum_{i,j} h_{ij}\,G^{\rho}_{ij} ~+~ \sum_n\log\left(2\cosh \epsilon^h_n\right)~. 
    \end{split}
\end{equation}
Thus the next step is to express $h_ij$ in terms of $G^{\rho}$, which can be derived via direct computations. It turns out that they share the same orthogonal matrix in their Schur decompositions and entries in their skew-symmetric matrices are related as follows
\begin{equation}
    \epsilon^h_n ~=~ {\rm arctanh}\left( - \epsilon^{G^{\rho}}_n\right) ~. 
    \label{equ:epsM}
\end{equation}
After some algebra, the final result of $S(\rho)$ reads
\begin{equation}
    \begin{split}
        &S(\rho) ~=~ \sum_n \epsilon^h_n\,\epsilon^{G^{\rho}}_n ~+~ \sum_n\log\left(2\cosh \epsilon^h_n\right)\\
        =& -\frac{1+\epsilon^{G^{\rho}}_n}{2}\log\frac{1+\epsilon^{G^{\rho}}_n}{2} - \frac{1-\epsilon^{G^{\rho}}_n}{2}\log\frac{1-\epsilon^{G^{\rho}}_n}{2}~.
    \end{split}
\end{equation}
This result was first derived in Ref.~\cite{Vidal_2003}. 

\subsection{Petz Map and Fidelity}
\label{app:petz}
In this section, we will apply the formalism introduced above to compute the Petz recovery map for Gaussian channels and the explicit expression of the fidelity between two Gaussian states. This subsection is essentially a review of the results in Ref.~\cite{Swingle:2018dto}.

Before diving into this computation, we first look at how correlation matrices transform explicitly under two basic operations on density matrices: partial trace over a subregion and tensor product. 
\begin{enumerate}
    \item Partial trace. Given a correlation matrix $(G^{\rho})_{IJ}$ ($I,J=1,\dots,2L$) for the state $\rho$ with $2L$ Majorana fermions, consider erasing a subset $A$ of these fermions. The correlation matrix for the reduced density matrix $\rho_{\bar{A}}=\tr_A\rho$ is obtained by deleting the rows and columns in $G^{\rho}$ corresponding to the fermions in $A$. 
    \item Tensor product. Following the same idea we can see that the correlation matrix for $\rho_A\otimes \rho_B$ is 
\begin{equation}
    G^{\rho_A\otimes \rho_B} ~=~ G^{\rho_A} \oplus G^{\rho_B} ~~.
    \label{equ:G-tensorprod}
\end{equation}
One useful example would be $G^{\rho_A\otimes \mathbb{I}_B} = G^{\rho_A}\oplus 0_{B}$.
\end{enumerate}

Moreover, in the rest of this appendix, we will adopt the following notations: (1) $0_A$ is the zero matrix on region $A$; 
(2) $\mathbb{I}_{A}$ is the identity matrix on region $A$. 

The rotated Petz recovery channel is defined in Eq.~(\ref{equ:rotated_Petz}). We consider the input state $X$ and the channel ${\cal N}$ to be both fermionic Gaussian. 
As mentioned earlier, the channel ${\cal N}$ can be described by its Grassmann integral representation with matrices $\mA^{\cal N}$ and $\mB^{\cal N}$. 
Then, given the transformation of the state, i.e. equation Eq.~(\ref{equ:rotated_Petz}), one can use the Choi–Jamiolkowski isomorphism to derive its integral representation. In \cite{Swingle:2018dto}, they rewrote the Petz map $(t=0)$ as a convolution of three Gaussian maps. For each map, they computed its integral representation and ultimately combined them. Given the correlation matrix $G^{\sigma}$ and the channel $\mathcal{N}$, they obtain
\begin{equation}\resizebox{0.41\textwidth}{!}{$%
    \begin{aligned}
        {\cal B}^{\rm Petz} =&~ \sqrt{\mathbb{I} + (G^{\sigma})^2}\,\left(\mB^{\cal N}\right)^{\rm T}\, \left(\sqrt{\mathbb{I} + \left(G^{{\cal N}(\sigma)}\right)^2}\right)^{-1} \\
         {\cal A}^{\rm Petz} =&~ G^{\sigma} - {\cal B}^{\rm Petz}\, G^{{\cal N}(\sigma)} \,\left({\cal B}^{\rm Petz}\right)^{\rm T} ~.
    \end{aligned}$}
    \label{equ:Petz_matrices}
\end{equation}

For our purpose of reconstructing subregions as explained in Sec.~\ref{sec:Petz}, we take $X = \rho_{AB} \otimes \mathbb{I}_C$, the channel ${\cal N}$ being the erasure channel on $C$, $\sigma = \rho_A \otimes \rho_{BC}$, hence ${\cal N}(\sigma)= \rho_A \otimes \rho_B \otimes \mathbb{I}_C$. By using Eq.~(\ref{equ:G-tensorprod}), we obtain
\begin{equation}
   \begin{split}
       &\mB^{\cal N} ~=~  \mathbb{I}_{AB}\oplus 0_C ~~,\\
      & G^{\sigma} ~=~ G^{\rho_A}\oplus G^{\rho_{BC}} ~~,\\
       &G^{{\cal N}(\sigma)} ~=~ G^{\rho_A}\oplus G^{\rho_{B}}\oplus 0_C ~~.
   \end{split}
\end{equation}

For generic $t$, the rotated Petz map is realized as follows
\begin{equation}
    \begin{split}
        {\cal D}_{\sigma,t}(X) ~=&~ \Big( {\cal U}_{\sigma,t}\, \circ\, {\cal D}_{\sigma,0}\, \circ\,{\cal U}_{{\cal N}(\sigma),-t}\Big)(X) ~, \\
        {\cal U}_{\sigma,t}(\cdot) ~=&~ \sigma^{\frac{it}{2}}\,(\cdot)\,\sigma^{-\frac{it}{2}} ~.
    \end{split}
\end{equation}
Based on the result of Eq.~(\ref{equ:Petz_matrices}), one can further derive the Grassmann integral representation of ${\cal D}_{\sigma,t}(X)$ and the final result is~\cite{Swingle:2018dto}
\begin{equation}\resizebox{0.41\textwidth}{!}{$%
    \begin{aligned}
        {\cal A}^{\rm rotated-Petz}(t) ~=&~ {\cal R}(G^{\sigma},t) \, A^{\rm Petz} \, \left( {\cal R}(G^{\sigma},t) \right)^{\rm T} ~,\\
        {\cal B}^{\rm rotated-Petz}(t) ~=&~ {\cal R}(G^{\sigma},t) \, {\cal B}^{\rm Petz} \, {\cal R}(G^{{\cal N}(\sigma)},-t) ~,
    \end{aligned}$}
\end{equation}
where ${\cal R}$ is defined as follows. 
Consider an antisymmetric matrix $M$, which admits the following decomposition
\begin{equation}
    M ~=~ O\,\left(\oplus_{i=1}^L\,
    \begin{pmatrix}
        0 & \epsilon_i\\
        -\epsilon_i & 0
    \end{pmatrix}\right)\,O^{\rm T} ~~.
\end{equation}
Then ${\cal R}(M,t)$ reads 
\begin{equation}
    \begin{aligned}
        &{\cal R}(M,t) ~:=~ 
        O\,\left(\oplus_{i=1}^L\, M_{i,t} \right)\,O^{\rm T} ~,\\
        & M_{i,t} = \begin{pmatrix}
            {\rm real}\left(\left(\frac{1+\epsilon_i}{1-\epsilon_i}\right)^{it/2}\right) & -{\rm img}\left(\left(\frac{1+\epsilon_i}{1-\epsilon_i}\right)^{it/2}\right) \\
            {\rm img}\left(\left(\frac{1+\epsilon_i}{1-\epsilon_i}\right)^{it/2}\right) & {\rm real}\left(\left(\frac{1+\epsilon_i}{1-\epsilon_i}\right)^{it/2}\right)
        \end{pmatrix}~,
    \end{aligned}
\end{equation}
where ${\rm real}(x)$ and ${\rm img}(x)$ denotes the real and imaginary part of $x$ respectively. 
Note that ${\cal R}(M,t=0)=\mathbb{I}$. Namely, when $t=0$, it reduces to the Petz map Eq.~(\ref{equ:Petz_matrices}) as expected.  
Moreover, when $\epsilon_i \approx \pm 1$, the matrix $M_{i,t}$ is singular. In this case, the $i$-th fermion decouples thus we set $M_{i,t} := \mathbb{I}_2$.

Finally, by using Eq.~(\ref{equ:Cij_transf_TP}), the correlation matrix for $\tilde{\rho}_{ABC}(t)$ defined in Eq.~(\ref{eq:Petz_map_1}) can be straightforwardly computed as follows
\begin{equation}\resizebox{0.41\textwidth}{!}{$%
    \begin{aligned}
        &G^{\tilde{\rho}_{ABC}(t)} ~=~ {\cal A}^{\rm rotated-Petz}(t) ~+~ \\
        & {\cal B}^{\rm rotated-Petz}(t) \Big( G^{\rho_{AB}}\oplus 0_{C} \Big)\left({\cal B}^{\rm rotated-Petz}(t)\right)^{\rm T} .
    \end{aligned}$}
\end{equation}

Now let's move on to fidelity. 
The fidelity of two states $\rho$ and $\sigma$ is defined as
\begin{equation}\label{equ:fidelity_def}
    F(\rho,\sigma) ~:=~ \tr\left( \sqrt{\sigma^{\frac{1}{2}}\,\rho\,\sigma^{\frac{1}{2}}} \right) ~.
\end{equation}
We would like to express $F(\rho,\sigma)$ in terms of the correlation matrices $G^{\rho}$ and $G^{\sigma}$ and the idea is the following. First, direct computation yields the Grassmann representation of $\sigma^{\frac{1}{2}}\rho\,\sigma^{\frac{1}{2}}$ as 
\begin{equation}
    \omega(\sigma^{\frac{1}{2}}\rho\,\sigma^{\frac{1}{2}},\theta) ~=~ \mathbf{c}\,\exp\left( \frac{i}{2}\,\theta^{\rm T}\,G^{\sigma^{\frac{1}{2}}\rho\sigma^{\frac{1}{2}}}\,\theta \right)~,
\end{equation}
where 
\begin{equation}
    \begin{split}
        &G^{\sigma^{\frac{1}{2}}\rho\sigma^{\frac{1}{2}}} ~=~ G^{\sigma} ~+~ \sqrt{\mathbb{I}+(G^{\sigma})^2}\\
        &\quad\quad\quad ~\times\left( \left( G^{\rho}\right)^{-1} -  G^{\sigma}\right)^{-1}\,\sqrt{\mathbb{I}+(G^{\sigma})^2}~,
    \end{split}
\end{equation}
and the prefactor $\mathbf{c}$ reads 
\begin{equation}
    \mathbf{c} ~=~ \frac{(-1)^L}{2^{2L}}\,{\rm Pf}(G^{\rho})\,{\rm Pf}\left((G^{\rho})^{-1} - G^{\sigma}\right) ~.
\end{equation}
Second, one can check that the Grassmann representation of the square root of an operator takes the following form
\begin{equation}
    \begin{split}
        \omega(\sqrt{X},\theta) ~=&~ \sqrt{\mathbf{c}}\,\left({\rm det}\left(\mathbb{I} - \left(G^{\sqrt{X}}\right)^2 \right)\right)^{-\frac{1}{4}}\\
        &\qquad\qquad ~\times~\exp\left( \frac{i}{2}\,\theta^{\rm T}\,G^{\sqrt{X}}\,\theta \right)~,
    \end{split}
\end{equation}
where 
\begin{equation}
    G^{\sqrt{X}} ~=~ \Big( G^X \Big)^{-1}\,\left( \sqrt{\mathbb{I} + \left(G^X\right)^2} - \mathbb{I} \right) ~.
\end{equation}
Third, note that the trace of an operator $X\in {\cal C}_{2L}$ can be written as 
\begin{equation}
    \tr(X) ~=~ 2^L\,\omega(X,\theta=0) ~.
\end{equation} This can be seen from the fact that a product of any nonzero number of \textit{distinct} Majorana fermions has a vanishing trace. 
Put it all together and we have
\begin{equation}
    \begin{split}
        F(\rho,&\sigma) ~=~ \frac{1}{2^{L/2}}\, \left({\rm det}\left(\mathbb{I}-G^{\rho}G^{\sigma}\right)\right)^{1/4} \\
        &~\times~\left({\rm det}\left(\mathbb{I}+\sqrt{\mathbb{I}+(G^{\sigma^{\frac{1}{2}}\rho\sigma^{\frac{1}{2}}})^2 }\right)\right)^{1/4} ~.
    \end{split}
\end{equation}

\subsection{Ising CFT States and Chiral States}\label{appen:CFT}
One example of systems that can be described by Eq.~(\ref{equ:H}) is the Ising model. To study the Ising CFT states, here we focus on its critical point
\begin{equation}
    H_{\rm critical~Ising} ~=~ -\,\sum_{i=1}^L\left( X_i\,X_{i+1} + Z_i \right) ~,
    \label{equ:H_Ising}
\end{equation}
where we consider the system size to be $L$ and $X$, $Z$ are Pauli operators. 
Indeed, Eq.~(\ref{equ:H_Ising}) takes the form of Eq.~(\ref{equ:H}) and one can see this by using the Jordan-Wigner transformation
\begin{equation}
    \gamma_{2i-1} ~=~ \left( \prod_{j<i}Z_j \right)\,X_i ~,~ 
    \gamma_{2i} ~=~ \left( \prod_{j<i}Z_j \right)\,Y_i ~.
\end{equation}
Then direct computation yields 
\begin{equation}
    Z_i ~=~ -i\,\gamma_{2i-1}\gamma_{2i}~,~ 
    X_iX_{i+1} ~=~ -i\,\g_{2i}\g_{2i+1} ~,
    \label{equ:Jordan-Wigner}
\end{equation}
and the Hamiltonian Eq.~(\ref{equ:H_Ising}) becomes 
\begin{equation}
    H_{\rm critical~Ising}~=~\frac{i}{2}\sum_{k=1}^{2L}\left( \g_k\g_{k+1} - \g_{k+1}\g_k \right) ~~. 
    \label{equ:H_Ising(gamma)}
\end{equation}
The periodic boundary condition $X_{L+1}=X_1$ corresponds to $\g_{2L+1} = -\g_1$ in the even spin parity sector. 

Therefore, Ising CFT states are fermionic Gaussian states and we can harness the formalism developed above to obtain results in section~\ref{sec:CFT}. Here we focus on the ground state and derive the correlation matrix for the ground state of a generic Hamiltonian Eq.~(\ref{equ:H}). As mentioned above, $H$ can be decomposed as follows
\begin{equation}
    \begin{split}
        H ~=&~ \frac{i}{2}\,\sum_{i,j=1}^{2L}\sum_{n,m=1}^{2L}\,\gamma_i\,O^M_{i,n}\,\Lambda^M_{nm}\,O^M_{j,m}\,\gamma_j  \\
        ~=&~i\,\sum_{k=1}^{L}\,\epsilon^M_k\,\tilde{\gamma}_{2k-1}\,\tilde{\gamma}_{2k}  ~~,
    \end{split}
\end{equation}
where we've noted that
\begin{equation}
    \Lambda^M_{nm} ~=~ \epsilon^M_k\delta_{n,2k-1}\delta_{m,2k} - \epsilon^M_k\delta_{n,2k}\delta_{m,2k-1} ~,
\end{equation}
and defined 
\begin{equation}
    \tilde{\gamma}_{2k-1} ~:=~ \sum_{i=1}^{2L} \gamma_i\,O^M_{i,2k-1} ~,~ 
    \tilde{\gamma}_{2k} ~:=~ \sum_{j=1}^{2L} \gamma_j\,O^M_{j,2k} ~.
    \label{equ:tilde-gamma}
\end{equation}
Using Eq.~(\ref{equ:Jordan-Wigner}) the Hamiltonian can be further written as 
\begin{equation}
    H ~=~ -\,\sum_{k=1}^{L}\,\epsilon^M_k\,\tilde{Z}_k ~.
\end{equation}
Thus in the $\tilde{Z}$ basis, the ground state has the following configuration: if $\epsilon^M_k$ is positive, spin $k$ is up (eigenstate of $\tilde{Z}_k$ with eigenvalue $+1$) and vice versa. Namely, 
\begin{equation}
    |0\rangle ~=~ |{\rm sgn}(\epsilon^M_1),\,\dots,\, {\rm sgn}(\epsilon^M_L)\rangle ~.
\end{equation}
With this explicit expression of the ground state in hand, we are able to compute a lot of quantities straightforwardly. 
One immediate result is the ground state energy,  
\begin{equation}
    E_{GS} ~=~ -\,\sum_{k=1}^{L}\,|\epsilon^M_k| ~.
\end{equation}
Moreover, the correlation matrix in $\tilde{\gamma}$ basis reads
\begin{equation}
    \begin{split}
        &\tilde{G}_{ij} ~=~ \langle 0| i\,\tilde{\gamma}_i\tilde{\gamma}_j |0\rangle\\
        ~=&~ \sum_{k=1}^L\Big( \delta_{i,2k}\delta_{j,2k-1} -\,\delta_{i,2k-1}\delta_{j,2k}\Big)\, \langle 0|\tilde{Z}_k |0\rangle \\
        ~=&~ \sum_{k=1}^L\Big( \delta_{i,2k}\delta_{j,2k-1} -\,\delta_{i,2k-1}\delta_{j,2k}\Big)\,{\rm sgn}(\epsilon^M_k) ~.
    \end{split}
\end{equation}
Finally, the correlation matrix in $\gamma$ basis can be obtained by implementing the inverse of the transformation Eq.~(\ref{equ:tilde-gamma}). Namely, 
\begin{equation}
    \begin{split}
       {G}_{ij} ~=&~  \langle 0| i\,{\gamma}_i{\gamma}_j |0\rangle = \langle 0| i\,\sum_{k,l} \tilde{\gamma}_k\,O^M_{i,k}\,\tilde{\gamma}_l\,O^M_{j,l}  |0\rangle\\
       ~=&~ \sum_{k,l}  O^M_{i,k}\, \tilde{G}_{kl} \,O^M_{j,l} ~.
    \end{split}
\end{equation}

It's worth mentioning that a shortcut to derive the correlation matrix for the Ising CFT ground state is through the discrete Fourier transform due to the translation invariance, i.e. $\gamma_{k}\mapsto\gamma_{k+1}$ in terms of the Majorana fermion. 
Note that $\gamma_k$ is located on the lattice and its Fourier mode reads
\begin{equation}
    \gamma(p) ~=~ \frac{1}{\sqrt{2L}}\sum_{k=1}^{2L}\gamma_k\,e^{ipk} ~.
\end{equation}
The hermiticity of $\gamma_k$ yields 
\begin{equation}
    \gamma^{\dagger}(p) ~=~ \gamma(-p) ~.
\end{equation}
With the translation invariance, the Neveu-Schwarz boundary condition becomes $\gamma_{2L+k} = -\gamma_k$. Therefore, $e^{ip2L}=-1$ and $p$ also takes discrete values as 
\begin{equation}
    p_n ~=~ \frac{\pi}{L}\left(n+\frac{1}{2}\right) ~,
\end{equation}
where $n\in[-L,L-1]$. The anticommutator relation of $\gamma_j$ implies 
\begin{equation}
    \{ \gamma(-p_n),\gamma(p_m)\} ~=~ 2\,\delta_{mn} ~,
    \label{equ:gamma(p)-commutator}
\end{equation}
where we have used the identity
\begin{equation}
    \sum_{k}\,e^{ipk} ~=~ 2\pi\,\delta_p ~.
\end{equation}
Note that the inverse of the discrete Fourier transform is
\begin{equation}
    \gamma_k ~=~ \frac{1}{\sqrt{2L}}\sum_{n=-L}^{L-1}\gamma(p_n)\,e^{-ip_nk} ~.
\end{equation} 
Plugging in Eq.~(\ref{equ:H_Ising(gamma)}), the Hamiltonian reads
\begin{equation}
    H_{\rm critical~Ising} ~=~ \sum_{n=0}^{L-1}\,\sin p_n\,\Big[ 2\,\gamma^{\dagger}(p_n)\gamma(p_n) - 1  \Big]~.
    \label{equ:H(p)}
\end{equation}

Given Eq.~(\ref{equ:gamma(p)-commutator}) and Eq.~(\ref{equ:H(p)}), when the momentum $p$ is positive, $\gamma(-p)=\gamma^{\dagger}(p)$ can be viewed as the creation operator while $\gamma(p)$ annihilates the vacuum. 
Thus in the momentum space, the two-point correlation function is simply the delta function 
\begin{equation}
    \langle 0| \gamma(p_n)\,\gamma(-p_m)|0\rangle ~=~ 2\,\delta_{nm}\,\Theta(p_n) ~,
    \label{equ:2pt-Ising}
\end{equation}
where $\Theta(p_n)$ is the Heaviside step function. 
The correlation matrix Eq.~(\ref{equ:Cij-def}) can be obtained via the inverse Fourier transform as follows
\begin{equation}
    \begin{split}
        & G^{\rm Ising}_{jl} ~=~ \langle 0| i\,\gamma_j\gamma_l|0\rangle \\
        ~=&~ \frac{i}{2L}\,\sum_{n,m}\langle 0| \gamma(p_n)\,\gamma(p_m)|0\rangle\,e^{-i(p_nj+p_ml)} \\
        ~=&~ \frac{i}{L}\,\sum_{n=0}^{L-1}e^{-i\frac{\pi}{L}(n+\frac{1}{2})(j-l)} \\
        ~=&~ \begin{cases}
        0 \qquad\qquad\qquad\qquad\quad j-l~\text{is even} \\
        1/\left[L\, \sin\frac{\pi}{2L}(j-l) \right] ~~~ j-l~\text{is odd} \\
    \end{cases}~~.
    \end{split}
    \label{equ:Ising_Cij}
\end{equation}

Next, we move on to the chiral thermal state. 
Note that in Eq.~(\ref{equ:H(p)}), $p_n\in (0,\pi)$, $\sin p_n >0 $, $\gamma^{\dagger}(p)\gamma(p)$ plays the role of the number operator and the ground state corresponds to zero occupation for all $p_n>0$. 
Based on the factor $\sin p_n$, momentum modes $p_n$ are naturally divided by $\pi/2$ into left-moving ($p_n>\pi/2$) and right-moving ($p_n<\pi/2$).  
Then the chiral state $|\psi\rangle_{c}$ can be constructed as follows. 
The left-moving modes stay in the ground state $(\beta_L=+\infty)$ while the right-moving modes are in the thermal state $(\beta_R=O(1))$. 
Namely, 
\begin{equation}
    \begin{split}
        &\gamma^{\dagger}(p_n)\gamma(p_n) |\psi\rangle_{c} ~=~ \frac{2e^{-\beta_R\sin p_n}}{1+e^{-\beta_R\sin p_n}}|\psi\rangle_{c} ~,~ n< \frac{L}{2}~,\\
        &\gamma^{\dagger}(p_n)\gamma(p_n) |\psi\rangle_{c} ~=~ 0 ~,~ n\ge \frac{L}{2} ~.\\
    \end{split}
\end{equation}

It follows that the two-point function in the momentum space is
\begin{equation}
    \begin{split}
        _c\langle \psi| \gamma(p_n)\,\gamma(-p_m)|\psi\rangle_c ~=&~ 2\,\delta_{nm}\,\phi_n(\beta_R) ~,\\
        _c\langle \psi| \gamma(-p_m)\,\gamma(p_n)|\psi\rangle_c ~=&~ 2\,\delta_{nm}\,\Big[ 1-\phi_n(\beta_R) \Big]~,
    \end{split}
    \label{equ:2pt-chiral}
\end{equation}
where $p_n, p_m>0$ and $\phi_n(\beta_R)$ takes the following form 
\begin{equation}
    \phi_n(\beta_R) = \begin{cases}
        1 \qquad\qquad\quad~~~\, n=\frac{L}{2},\cdots,L-1\\
        \frac{1}{1+e^{-\beta_R\sin p_n}} \quad n=0,\cdots,\frac{L}{2}-1
    \end{cases}.
    \label{equ:phi_chiral}
\end{equation}
Note that when $T_R\to 0$ or equivalently $\beta_R\to \infty$, $\phi_n$ becomes a constant one and Eq.~(\ref{equ:2pt-chiral}) reduces to Eq.~(\ref{equ:2pt-Ising}). 
The correlation matrix can be derived in a similar manner as Eq.~(\ref{equ:Ising_Cij}) and reads 
\begin{equation}
    \begin{split}
        G^{\rm chiral}_{jl}(\beta_R) ~=&~ -\,G_{jl}^{\rm Ising} 
        ~+~ \frac{2}{L}\,\sum_{n=0}^{L-1}\,\phi_n(\beta_R)\\
        &\quad~\times~\sin\left[ \frac{\pi}{L}\left(n+\frac{1}{2}\right)(j-l) \right]~~.
    \end{split}
\end{equation}

\subsection{Gaussian Measurements}\label{appen:measurement}
In this section, we consider the $Z$ measurement at site $n$, which is a Gaussian operation. The derivation follows from Sec.~IX of Ref.~\cite{bravyi2004lagrangian}.
Based on Eq.~(\ref{equ:Jordan-Wigner}), the projection operators for $Z$ measurement at site $n$ are
\begin{equation}
    P_{n,a} ~=~ \frac{1}{2}\,\Big( \mathbb{I} + i(-1)^a\,\gamma_{2n-1}\gamma_{2n} \Big) ~,~ a=0,1~.
\end{equation}
After the measurement, the state becomes one of the projected states with probabilities
\begin{equation}
    {\rm prob}(a) = \tr(\rho P_{n,a}) = \frac{1}{2}+\frac{(-1)^a}{2}\,(G^{\rho})_{2n-1,2n}~.
    \label{equ:prob(a)}
\end{equation}

In what follows, we derive the Grassmann representation of this projection.  
To do so, following the Choi–Jamiolkowski isomorphism, we first look at how the maximally entangled state $\rho_I=(1/2^{2L})\prod_{i=1}^{2L}(\mathbb{I}+i\gamma_i\gamma_{2L+i})$ transforms under the projections. 
They are
\begin{equation}
    \begin{split}
        & \rho_{n,a} = P_{n,a}\,\rho_I\, P_{n,a}
        = \frac{1}{2^{2L+1}}(\mathbb{I}+(-1)^ai\gamma_{2n-1}\gamma_{2n})\\
        &(\mathbb{I}-(-1)^ai\gamma_{2L+2n-1}\gamma_{2L+2n})\prod_{i\ne 2n-1,2n}(\mathbb{I}+i\gamma_i\gamma_{2L+i}).
    \end{split}
\end{equation}
Then the Grassmann representation of $\rho_{n,a}$ can be obtained by simply replacing $\gamma_i$ ($i\le 2L)$ with $\theta_i$ while $\gamma_i$ ($i> 2L)$ is replaced by $\eta_{i-2L}$. Namely, 
\begin{equation}
    \begin{split}
        &\rho_{n,a}(\theta,\eta)~=~ \frac{1}{2^{2L+1}}\,(\mathbb{I}+(-1)^ai\theta_{2n-1}\theta_{2n})\\
        \times&(\mathbb{I}-(-1)^ai\eta_{2n-1}\eta_{2n})\prod_{i\ne 2n-1,2n}(\mathbb{I}+i\theta_i\eta_{i})\\
        &\qquad\qquad ~=~ \frac{1}{2^{2L+1}}\exp\Big[S_{n,a}(\theta,\eta)\Big] ~,
    \end{split}
\end{equation}
where $S_{n,a}(\theta,\eta)$ takes the following form
\begin{equation}
    \begin{split}
        S_{n,a}(\theta,\eta) ~=&~ \Big[ i(-1)^a(\theta_{2n-1}\theta_{2n} - \eta_{2n-1}\eta_{2n})\\
        &\qquad\qquad\qquad + i\sum_{i\ne 2n-1,2n}\theta_i\eta_{i} \Big]~.
    \end{split}
    \label{equ:S_na}
\end{equation}

Based on Eq.~(\ref{equ:rho_E-rep}), the projection operation has the following Grassmann representation
\begin{equation}
    \begin{split}
        \frac{1}{2}\int D\mu D\eta\, \exp\Big[ S_{n,a}(\theta,\eta) \,+\, i\eta^{\rm T}\mu \Big]\omega(\rho,\mu)~.
    \end{split}
\end{equation}
Given the explicit expression of $S_{n,a}(\theta,\eta)$, after the projection $P_{n,a}$, the correlation matrix becomes $G^{\rho_{n,a}}$ and can be obtained straightforwardly by using Eq.~(\ref{equ:S(theta,eta)}) and Eq.~(\ref{equ:Cij_transf}). The explicit expression reads~\cite{bravyi2004lagrangian}
\begin{equation}
    \begin{split}
        G^{\rho_{n,a}} ~=~ (-1)^aK 
        + B\left(\left(G^{\rho}\right)^{-1} - (-1)^aK\right)^{-1}B ~.
    \end{split}
    \label{equ:proj-G}
\end{equation}
To further simplify the expression, we view a $2L\times 2L$ matrix $M$ as a $L\times L$ block matrix where each block $M_{p,q}$ is a $2\times 2$ matrix. Namely, 
\begin{equation}
    M ~=~ \begin{pmatrix}
       M_{1,1} & M_{1,2} & \cdots & M_{1,L} \\
       M_{2,1} & M_{2,2} & \cdots & M_{2,L} \\
        \vdots & \vdots & \vdots & \vdots\\
       M_{L,1} & M_{L,2} & \cdots & M_{L,L} \\
    \end{pmatrix}~. 
    \label{equ:M-block}
\end{equation}
The only nonzero blocks in $K$ and $B$ are
\begin{equation}
    K_{n,n} ~=~ \begin{pmatrix}
        0 & 1\\
        -1 & 0 
    \end{pmatrix} ~,~
\end{equation} and 
\begin{equation}
    B_{p,p} ~=~ \begin{pmatrix}
        1 & 0\\
        0 & 1 
    \end{pmatrix}  ~,~ p\ne n~.
\end{equation}

For $G^{\rho_{n,a}}$, it can also be written in the form of Eq.~(\ref{equ:M-block}). 
After an explicit computation of the matrix inverse in Eq.~(\ref{equ:proj-G}), $2\times 2$ matrices $G^{\rho_{n,a}}_{p,q}$ take the following form 
\begin{equation}
\begin{split}
    G^{\rho_{n,a}}_{p,q} ~=~ G^{\rho}_{p,q} + \frac{(-1)^{a}}{1+x}\,G^{\rho}_{p,n}\begin{pmatrix}
        0 & 1\\
        -1 & 0
    \end{pmatrix} G^{\rho}_{n,q}~,
\end{split}
\end{equation}
where both $p,q\ne n$ and $x = (-1)^a(G^{\rho})_{2n-1,2n}$. If one of $p$ or $q$ equals $n$, 
\begin{equation}
    G^{\rho_{n,a}}_{p,q} ~=~ (-1)^a \delta_{pn} \delta_{qn}\begin{pmatrix}
        0 & 1\\
        -1 & 0
    \end{pmatrix} ~.
\end{equation}

\section{Computing the Petz-map Fidelity with Exact Diagonalization}\label{appen:petz_via_diag}

This appendix explains our numerical method to compute the Petz map fidelity for MIPTs with exact diagonalization.  
We start with a pure state $|\psi_{ABCD}\rangle$ prepared by the circuit in Fig.~\ref{fig:mipt}, and the spin chain is divided into four contiguous regions $A$, $B$, $C$, and $D$ as follows
\begin{equation}
    \begin{tikzpicture}[baseline={([yshift=-0.9ex]current bounding box.center)},scale=0.7] 
    \draw [thick] (-5,0) -- (5,0);
    \draw [thick] (-5,0.2) -- (-5,-0.2);
    \draw (-4,0)node[below]{$A$};
    \draw [thick] (-3,0.2) -- (-3,-0.2);
    \draw (-1,0)node[below]{$B$};
    \draw [thick] (1,0.2) -- (1,-0.2);
    \draw (2,0)node[below]{$C$};
    \draw [thick] (3,0.2) -- (3,-0.2);
    \draw (4,0)node[below]{$D$};
    \draw [thick] (5,0.2) -- (5,-0.2);
\end{tikzpicture} ~~.
\label{fig:ABCD}
\end{equation}
The density matrix $\rho_{ABC}$ is 
\begin{equation}
    \rho_{ABC} ~=~ {\rm Tr}_D |\psi_{ABCD}\rangle\langle \psi_{ABCD}|  ~~.
\end{equation}

\paragraph{Uhlmann fidelity} 
In order to compute the Petz map fidelity in an efficient manner, in our numerical algorithm, we do not compute the recovered density matrix $\tilde{\rho}_{ABC}$ via Eq.~(\ref{eq:Petz_map_1}) explicitly and then plug in the definition Eq.~(\ref{equ:fidelity_def}) 
Instead, we derive $F_t$ using the Uhlmann theorem. This algorithm follows the appendix A of~\cite{Vardhan:2023pnm}. In a nutshell, we compute the Petz map fidelity by summing over all singular values of the matrix in Fig.~\ref{fig:Flambda_final}.

The idea is the following. First, the Uhlmann theorem shows that the fidelity of two density matrices is the maximal possible overlap of their purifications. Consider two arbitrary purifications $|\Phi_{\rho}\rangle$ and $|\Phi_{\sigma}\rangle$ for $\rho$ and $\sigma$ respectively. 
A generic purification takes the following form (here schematically we say $\rho$ is defined on region $A$)
\begin{equation}
    \begin{split}
        |{\Phi}_{\rho}(V)\rangle ~=&~ (\mathbb{I}_A \otimes V)\,|\Phi_{\rho}\rangle ~,\\
        |{\Phi}_{\sigma}(W)\rangle ~=&~ (\mathbb{I}_A \otimes W)\,|\Phi_{\rho}\rangle ~,\\
    \end{split}
\end{equation}
where $V$ and $W$ are generally isometries acting on the ancilla systems. 
Then the Uhlmann fidelity can be written as
\begin{equation}
    \begin{split}
        &F(\rho, \sigma) ~=~ \max_{V,W}\,\langle \Phi_{\rho}| (\mathbb{I} \otimes V^{\dagger}W) | \Phi_{\sigma}\rangle \\
        ~=&~ \max_{V,W}\,\begin{tikzpicture}[baseline={([yshift=-0.9ex]current bounding box.center)},scale=0.9]
    \node[draw, shape=rectangle] (v2) at (0.5,1) {$V^{\dagger}W$};
    \node[ellipse,minimum width=6em,draw] (v1) at (0,2.2) {$\Phi^{\dagger}_{\rho}$};
    \node[ellipse,minimum width=6em,draw] (v3) at (0,-0.3) {$\Phi_{\sigma}$};
    \draw [thick] (-0.5,0.1)node[above left]{\footnotesize $A$} -- (-0.5,1.7); 
    \draw [thick] (0.5,1.3) -- (0.5,1.7);
    \draw [thick] (0.5,0.1) -- (0.5,0.7); 
\end{tikzpicture} 
~=~ \sum_{\rm s.v.}\,\begin{tikzpicture}[baseline={([yshift=-0.9ex]current bounding box.center)},scale=0.9]
    \node[ellipse,minimum width=6em,draw] (v1) at (0,1.5) {$\Phi^{\dagger}_{\sigma}$};
    \node[ellipse,minimum width=6em,draw] (v3) at (0,-0.3) {$\Phi_{\rho}$};
    \draw [thick] (0,2) -- (0,2.55)node[right]{\footnotesize $B'$}; 
    \draw [thick] (0,-0.75) -- (0,-1.3)node[above left]{\footnotesize $B$};
    \draw [thick] (0,0.15)node[above left]{\footnotesize $A$} -- (0,1); 
\end{tikzpicture} ~,
    \end{split}
    \label{equ:Fexp}
\end{equation}
where ${\rm s.v.}$ stands for singular values. Note that in the second equality, the indices on $A$ are contracted, and the resulting matrix has indices on $B'$ and $B$.

\paragraph{Purifications}
Given Eq.~(\ref{equ:Fexp}), now the task is to find efficient purifications for $\rho_{ABC}$ and $\tilde{\rho}_{ABC}$. Note that to purify $\rho$, the size of the ancilla has to be at least the rank of $\rho$. 
Therefore, our strategy is as follows. 
For $\rho_{ABC}$, if $L_A+L_B+L_C\ge L/2$, we use the original pure state $|\Phi_{\rho_{ABC}}\rangle=|\psi_{ABCD}\rangle$, which is obviously a purification of $\rho_{ABC}$. When $L_A+L_B+L_C< L/2$, we choose the canonical purification $|\Phi_{\rho_{ABC}}\rangle=|\sqrt{\rho_{ABC}}\rangle$ defined as follows. For a generic density matrix $\rho$, its square root can be computed via diagonalization
\begin{equation}
    \rho = \sum_a\,p_a\,|\psi_a\rangle\langle \psi_a| ~\rightarrow~ \sqrt{\rho} = \sum_a\,\sqrt{p_a}\,|\psi_a\rangle\langle \psi_a|~.
\end{equation}
One can check that the following state is a purification of $\rho$ which is known as the canonical purification
\begin{equation}
    |\sqrt{\rho}\rangle ~=~ \sum_a\,\sqrt{p_a}\,|\psi_a\rangle|\psi_a\rangle ~~.
\end{equation}
This procedure also applies to $\rho_{AB}$. Namely,
\begin{equation}
   |\Phi_{\rho_{AB}}\rangle ~=~ 
   \begin{cases}
      |\sqrt{\rho_{AB}}\rangle \qquad~\, L_A+L_B<L/2 ~~,\\
     |\psi_{ABCD}\rangle \qquad L_A+L_B\ge L/2~~.
   \end{cases}
\end{equation}
For $\tilde{\rho}_{ABC}$, one can check the following state is its purification
\begin{equation}
\begin{split}
    |\Phi_{\tilde{\rho}_{ABC}}\rangle ~=&~ \rho_{BC}^{\frac{1}{2}+i\frac{t}{2}}\,\rho_B^{-\frac{1}{2}-i\frac{t}{2}}\,|\Phi_{\rho_{AB}}\rangle \\
    ~=&~ 
    \begin{tikzpicture}[baseline={([yshift=-0.9ex]current bounding box.center)},scale=0.9]
    \draw [thick] (-0.5,0.1) -- (-0.5,3)node[above]{$A$};
    \draw [thick] (-0.5,-0.7) -- (-0.5,-1.2)node[below]{$C'$};
    \draw [thick] (0.5,-0.7) -- (0.5,-1.2)node[below]{$D$};
    \draw [thick] (0.5,2.6) -- (0.5,3)node[above]{$B$};
    \draw [thick] (1.5,2.6) -- (1.5,3)node[above]{$C$};
    \draw [thick] (1.5,1.8) -- (1.5,-1.2)node[below]{$C^*$};
    \node[draw, shape=rectangle] (v0) at (1.1,2.2) {~~~~$\rho_{BC}^{\frac{1}{2}+i\frac{t}{2}}$~~~~};
    \node[draw, shape=rectangle] (v1) at (0.5,1) {$\rho_B^{-\frac{1}{2}-i\frac{t}{2}}$};
    \node[ellipse,minimum width=6em,draw] (v3) at (0,-0.3) {$\Phi_{\rho_{AB}}$};
    \draw [thick] (v1) -- (v1 |- v3.north); 
    \draw [thick] (v1) -- (v1 |- v0.south);
\end{tikzpicture} ~~.
\end{split}
\end{equation}

\paragraph{Optimizations}
Based on Eq.~(\ref{equ:Fexp}), we have 
\begin{equation}
    F_t ~=~ \sum_{\text{s.v.}}\,\Phi^{\dagger}_{\rho_{ABC}}\,\Phi_{\tilde{\rho}_{ABC}} ~.
\end{equation}
Diagrammatically $\Phi^{\dagger}_{\rho_{ABC}}\,\Phi_{\tilde{\rho}_{ABC}}$ is shown in Fig.~\ref{fig:Flambda_final}, 
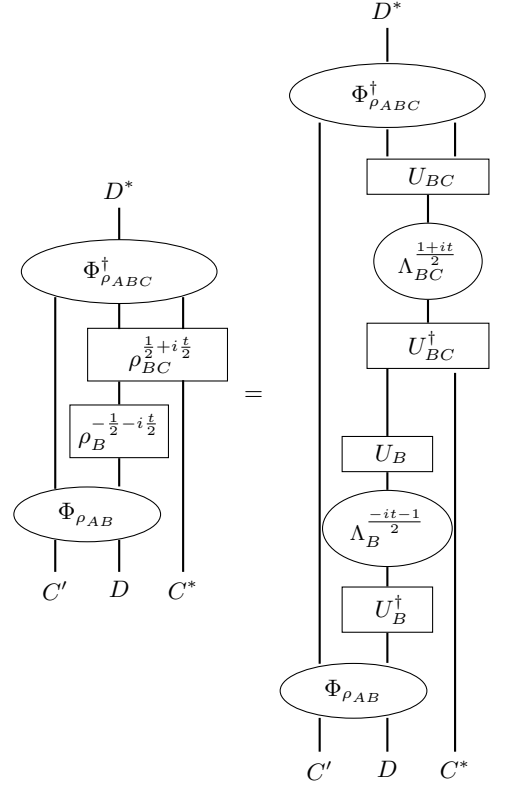
\begin{figure}[t]
    \centering
    \begin{tikzpicture}[baseline={([yshift=-0.9ex]current bounding box.center)},scale=0.85]
    \draw [thick] (-0.5,0.1) -- (-0.5,3.1);
    \draw [thick] (-0.5,-0.7) -- (-0.5,-1.2)node[below]{$C'$};
    \draw [thick] (0.5,-0.7) -- (0.5,-1.2)node[below]{$D$};
    \draw [thick] (0.5,2.6) -- (0.5,3);
    \draw [thick] (1.5,2.6) -- (1.5,3.1);
    \draw [thick] (1.5,1.8) -- (1.5,-1.2)node[below]{$C^*$};
    \draw [thick] (0.5,4) -- (0.5,4.5)node[above]{$D^*$};
    \node[ellipse,minimum width=8em,draw] (v2) at (0.5,3.5) {$\Phi^{\dagger}_{\rho_{ABC}}$};
    \node[draw, shape=rectangle] (v0) at (1.1,2.2) {~~~~$\rho_{BC}^{\frac{1}{2}+i\frac{t}{2}}$~~~~};
    \node[draw, shape=rectangle] (v1) at (0.5,1) {$\rho_B^{-\frac{1}{2}-i\frac{t}{2}}$};
    \node[ellipse,minimum width=6em,draw] (v3) at (0,-0.3) {$\Phi_{\rho_{AB}}$};
    \draw [thick] (v1) -- (v1 |- v3.north); 
    \draw [thick] (v1) -- (v1 |- v0.south);
\end{tikzpicture} 
%%%%%%%%%%%%%%%%%%%%%%%%%%%%%%%%%%%%%%%%%%%%%%
~=~ \begin{tikzpicture}[baseline={([yshift=-0.9ex]current bounding box.center)},scale=0.9]
    \draw [thick] (-0.5,0.1-3) -- (-0.5,3.1+2);
    \draw [thick] (-0.5,-0.7-3) -- (-0.5,-1.2-3)node[below]{$C'$};
    \draw [thick] (0.5,-0.7-3) -- (0.5,-1.2-3)node[below]{$D$};
    \draw [thick] (0.5,2.6+2) -- (0.5,3+2);
    \draw [thick] (1.5,2.6+2) -- (1.5,3.1+2);
    \draw [thick] (1.5,1.8-0.4) -- (1.5,-1.2-3)node[below]{$C^*$};
    \draw [thick] (0.5,4+2) -- (0.5,4.5+2)node[above]{$D^*$};
    \node[ellipse,minimum width=8em,draw] (v2) at (0.5,3.5+2) {$\Phi^{\dagger}_{\rho_{ABC}}$};
    \node[draw, shape=rectangle] (v01) at (1.1,2.3+2) {~~~~$U_{BC}$~~~~};
    \node[draw, shape=ellipse] (v02) at (1.1,2.2+0.85) {\footnotesize{$\Lambda_{BC}^{\frac{1+it}{2}}$}};
    \node[draw, shape=rectangle] (v03) at (1.1,2.2-0.4) {~~~~$U_{BC}^\dagger$~~~~};
    \node[draw, shape=rectangle] (v11) at (0.5,1-0.8) {~~~$U_B$~~~};
    \node[draw, shape=ellipse] (v12) at (0.5,1-0.4-1.5) {\footnotesize{$\Lambda_B^{\frac{-it-1}{2}}$}};
    \node[draw, shape=rectangle] (v13) at (0.5,1-0.4-2.7) {~~~$U^{\dagger}_B$~~~};
    \node[ellipse,minimum width=6em,draw] (v3) at (0,-0.3-3) {$\Phi_{\rho_{AB}}$};
    \draw [thick] (v02) -- (v02 |- v03.north);
    \draw [thick] (v02) -- (v02 |- v01.south);
    \draw [thick] (v13) -- (v13 |- v3.north);
    \draw [thick] (v13) -- (v13 |- v12.south);
    \draw [thick] (v11) -- (v11 |- v12.north);
    \draw [thick] (v11) -- (v11 |- v03.south);
\end{tikzpicture}
\caption{Matrix $\Phi^{\dagger}_{\rho_{ABC}}\,\Phi_{\tilde{\rho}_{ABC}}$ in the computation of the Uhlmann fidelity for the Petz recovery map. $\Phi_{\rho_{AB}}$ and $\Phi_{\rho_{ABC}}$ are two purifications for $\rho_{AB}$ and $\rho_{ABC}$ respectively. }
\label{fig:Flambda_final}
\end{figure}
where we have decomposed the density matrices $\rho_{BC}^{\frac{1}{2}+i\frac{t}{2}}$ and $\rho_B^{-\frac{1}{2}-i\frac{t}{2}}$ as follows 
\begin{equation}
    \begin{split}
        \rho_{BC}^{\frac{1}{2}+i\frac{t}{2}} ~=&~ U_{BC}\,\Lambda_{BC}^{\frac{1}{2}+i\frac{t}{2}}\,U^{\dagger}_{BC} ~~,\\
        \rho_B^{-\frac{1}{2}-i\frac{t}{2}} ~=&~ U_{B}\,\Lambda_{B}^{-\frac{1}{2}-i\frac{t}{2}}\,U^{\dagger}_{B} ~~.\\
    \end{split}
    \label{equ:rho_decomp}
\end{equation}
$\Lambda_{BC}$ and $\Lambda_{B}$ are diagonal matrices and their entries are eigenvalues of $\rho_{BC}$ and $\rho_{B}$ respectively. We note that the following two tricks help reduce the computational cost of the matrix in Fig.~\ref{fig:Flambda_final}. 
\begin{enumerate}
    \item In Eq.~\eqref{equ:rho_decomp}, the $U$'s and $\Lambda$'s are computed via SVD of $\psi_{ABCD}$ as a matrix if the subsystem size is larger than $L/2$; otherwise they are obtained by an eigenvalue decomposition of the densely computed density matrix. 
    \item The order of tensor contraction is important. When computing the matrix in Fig.~\ref{fig:Flambda_final}, we follow a strategy that minimizes the number of free indices at each step of the contraction. 
    Specifically, we adopt the following sequence for efficiency 
    \begin{equation}
        \begin{split}
            &\Big[\left(\Phi^{\dagger}_{\rho_{ABC}}\left(U_{BC}\Lambda_{BC}^{\frac{1}{2}+i\frac{t}{2}}\right)\right)U^{\dagger}_{BC}\Big]\\
        &\qquad ~\times~\Big[\left(U_{B}\Lambda_{B}^{-\frac{1}{2}-i\frac{t}{2}}\right)\left(U^{\dagger}_{B}\Phi_{\rho_{AB}}\right)\Big] ~,
        \end{split}
    \end{equation}
    for which the computational cost is $O(2^{3L_{ABC}})$. The maximum size of tensors that occurs in this computation is $2^L$. 
\end{enumerate}
In this work, we simulate MIPT states with $L=20$ qubits. For a given MIPT state, calculating the Petz fidelity in the most computationally intensive subregion configuration requires on the order of one second on a 10-core CPU with 2.4 GHz frequency.

\bibliography{reference}

%apsrev4-2.bst 2019-01-14 (MD) hand-edited version of apsrev4-1.bst
%Control: key (0)
%Control: author (8) initials jnrlst
%Control: editor formatted (1) identically to author
%Control: production of article title (0) allowed
%Control: page (0) single
%Control: year (1) truncated
%Control: production of eprint (0) enabled
\begin{thebibliography}{52}%
\makeatletter
\providecommand \@ifxundefined [1]{%
 \@ifx{#1\undefined}
}%
\providecommand \@ifnum [1]{%
 \ifnum #1\expandafter \@firstoftwo
 \else \expandafter \@secondoftwo
 \fi
}%
\providecommand \@ifx [1]{%
 \ifx #1\expandafter \@firstoftwo
 \else \expandafter \@secondoftwo
 \fi
}%
\providecommand \natexlab [1]{#1}%
\providecommand \enquote  [1]{``#1''}%
\providecommand \bibnamefont  [1]{#1}%
\providecommand \bibfnamefont [1]{#1}%
\providecommand \citenamefont [1]{#1}%
\providecommand \href@noop [0]{\@secondoftwo}%
\providecommand \href [0]{\begingroup \@sanitize@url \@href}%
\providecommand \@href[1]{\@@startlink{#1}\@@href}%
\providecommand \@@href[1]{\endgroup#1\@@endlink}%
\providecommand \@sanitize@url [0]{\catcode `\\12\catcode `\$12\catcode `\&12\catcode `\#12\catcode `\^12\catcode `\_12\catcode `\%12\relax}%
\providecommand \@@startlink[1]{}%
\providecommand \@@endlink[0]{}%
\providecommand \url  [0]{\begingroup\@sanitize@url \@url }%
\providecommand \@url [1]{\endgroup\@href {#1}{\urlprefix }}%
\providecommand \urlprefix  [0]{URL }%
\providecommand \Eprint [0]{\href }%
\providecommand \doibase [0]{https://doi.org/}%
\providecommand \selectlanguage [0]{\@gobble}%
\providecommand \bibinfo  [0]{\@secondoftwo}%
\providecommand \bibfield  [0]{\@secondoftwo}%
\providecommand \translation [1]{[#1]}%
\providecommand \BibitemOpen [0]{}%
\providecommand \bibitemStop [0]{}%
\providecommand \bibitemNoStop [0]{.\EOS\space}%
\providecommand \EOS [0]{\spacefactor3000\relax}%
\providecommand \BibitemShut  [1]{\csname bibitem#1\endcsname}%
\let\auto@bib@innerbib\@empty
%</preamble>
\bibitem [{\citenamefont {Chen}\ \emph {et~al.}(2010)\citenamefont {Chen}, \citenamefont {Gu},\ and\ \citenamefont {Wen}}]{Chen_2010}%
  \BibitemOpen
  \bibfield  {author} {\bibinfo {author} {\bibfnamefont {X.}~\bibnamefont {Chen}}, \bibinfo {author} {\bibfnamefont {Z.-C.}\ \bibnamefont {Gu}},\ and\ \bibinfo {author} {\bibfnamefont {X.-G.}\ \bibnamefont {Wen}},\ }\bibfield  {title} {\bibinfo {title} {Local unitary transformation, long-range quantum entanglement, wave function renormalization, and topological order},\ }\href {https://doi.org/10.1103/PhysRevB.82.155138} {\bibfield  {journal} {\bibinfo  {journal} {Phys. Rev. B}\ }\textbf {\bibinfo {volume} {82}},\ \bibinfo {pages} {155138} (\bibinfo {year} {2010})}\BibitemShut {NoStop}%
\bibitem [{\citenamefont {Wen}(2017)}]{Wen_Zoo_2017}%
  \BibitemOpen
  \bibfield  {author} {\bibinfo {author} {\bibfnamefont {X.-G.}\ \bibnamefont {Wen}},\ }\bibfield  {title} {\bibinfo {title} {Colloquium: Zoo of quantum-topological phases of matter},\ }\href {https://doi.org/10.1103/RevModPhys.89.041004} {\bibfield  {journal} {\bibinfo  {journal} {Rev. Mod. Phys.}\ }\textbf {\bibinfo {volume} {89}},\ \bibinfo {pages} {041004} (\bibinfo {year} {2017})}\BibitemShut {NoStop}%
\bibitem [{\citenamefont {Kitaev}\ and\ \citenamefont {Preskill}(2006)}]{Kitaev_2006}%
  \BibitemOpen
  \bibfield  {author} {\bibinfo {author} {\bibfnamefont {A.}~\bibnamefont {Kitaev}}\ and\ \bibinfo {author} {\bibfnamefont {J.}~\bibnamefont {Preskill}},\ }\bibfield  {title} {\bibinfo {title} {Topological entanglement entropy},\ }\bibfield  {journal} {\bibinfo  {journal} {Physical Review Letters}\ }\textbf {\bibinfo {volume} {96}},\ \href {https://doi.org/10.1103/physrevlett.96.110404} {10.1103/physrevlett.96.110404} (\bibinfo {year} {2006})\BibitemShut {NoStop}%
\bibitem [{\citenamefont {Levin}\ and\ \citenamefont {Wen}(2006)}]{Levin_2006}%
  \BibitemOpen
  \bibfield  {author} {\bibinfo {author} {\bibfnamefont {M.}~\bibnamefont {Levin}}\ and\ \bibinfo {author} {\bibfnamefont {X.-G.}\ \bibnamefont {Wen}},\ }\bibfield  {title} {\bibinfo {title} {Detecting topological order in a ground state wave function},\ }\bibfield  {journal} {\bibinfo  {journal} {Physical Review Letters}\ }\textbf {\bibinfo {volume} {96}},\ \href {https://doi.org/10.1103/physrevlett.96.110405} {10.1103/physrevlett.96.110405} (\bibinfo {year} {2006})\BibitemShut {NoStop}%
\bibitem [{\citenamefont {Vidal}\ \emph {et~al.}(2003)\citenamefont {Vidal}, \citenamefont {Latorre}, \citenamefont {Rico},\ and\ \citenamefont {Kitaev}}]{Vidal_2003}%
  \BibitemOpen
  \bibfield  {author} {\bibinfo {author} {\bibfnamefont {G.}~\bibnamefont {Vidal}}, \bibinfo {author} {\bibfnamefont {J.~I.}\ \bibnamefont {Latorre}}, \bibinfo {author} {\bibfnamefont {E.}~\bibnamefont {Rico}},\ and\ \bibinfo {author} {\bibfnamefont {A.}~\bibnamefont {Kitaev}},\ }\bibfield  {title} {\bibinfo {title} {Entanglement in quantum critical phenomena},\ }\bibfield  {journal} {\bibinfo  {journal} {Physical Review Letters}\ }\textbf {\bibinfo {volume} {90}},\ \href {https://doi.org/10.1103/physrevlett.90.227902} {10.1103/physrevlett.90.227902} (\bibinfo {year} {2003})\BibitemShut {NoStop}%
\bibitem [{\citenamefont {Calabrese}\ and\ \citenamefont {Cardy}(2004)}]{Calabrese_2004}%
  \BibitemOpen
  \bibfield  {author} {\bibinfo {author} {\bibfnamefont {P.}~\bibnamefont {Calabrese}}\ and\ \bibinfo {author} {\bibfnamefont {J.}~\bibnamefont {Cardy}},\ }\bibfield  {title} {\bibinfo {title} {Entanglement entropy and quantum field theory},\ }\href {https://doi.org/10.1088/1742-5468/2004/06/P06002} {\bibfield  {journal} {\bibinfo  {journal} {Journal of Statistical Mechanics: Theory and Experiment}\ }\textbf {\bibinfo {volume} {2004}},\ \bibinfo {pages} {P06002} (\bibinfo {year} {2004})}\BibitemShut {NoStop}%
\bibitem [{\citenamefont {Jian}\ \emph {et~al.}(2020)\citenamefont {Jian}, \citenamefont {You}, \citenamefont {Vasseur},\ and\ \citenamefont {Ludwig}}]{Jian:2019mny}%
  \BibitemOpen
  \bibfield  {author} {\bibinfo {author} {\bibfnamefont {C.-M.}\ \bibnamefont {Jian}}, \bibinfo {author} {\bibfnamefont {Y.-Z.}\ \bibnamefont {You}}, \bibinfo {author} {\bibfnamefont {R.}~\bibnamefont {Vasseur}},\ and\ \bibinfo {author} {\bibfnamefont {A.~W.~W.}\ \bibnamefont {Ludwig}},\ }\bibfield  {title} {\bibinfo {title} {{Measurement-induced criticality in random quantum circuits}},\ }\href {https://doi.org/10.1103/PhysRevB.101.104302} {\bibfield  {journal} {\bibinfo  {journal} {Phys. Rev. B}\ }\textbf {\bibinfo {volume} {101}},\ \bibinfo {pages} {104302} (\bibinfo {year} {2020})},\ \Eprint {https://arxiv.org/abs/1908.08051} {arXiv:1908.08051 [cond-mat.stat-mech]} \BibitemShut {NoStop}%
\bibitem [{\citenamefont {Choi}\ \emph {et~al.}(2020)\citenamefont {Choi}, \citenamefont {Bao}, \citenamefont {Qi},\ and\ \citenamefont {Altman}}]{Choi:2019nhg}%
  \BibitemOpen
  \bibfield  {author} {\bibinfo {author} {\bibfnamefont {S.}~\bibnamefont {Choi}}, \bibinfo {author} {\bibfnamefont {Y.}~\bibnamefont {Bao}}, \bibinfo {author} {\bibfnamefont {X.-L.}\ \bibnamefont {Qi}},\ and\ \bibinfo {author} {\bibfnamefont {E.}~\bibnamefont {Altman}},\ }\bibfield  {title} {\bibinfo {title} {{Quantum Error Correction in Scrambling Dynamics and Measurement-Induced Phase Transition}},\ }\href {https://doi.org/10.1103/PhysRevLett.125.030505} {\bibfield  {journal} {\bibinfo  {journal} {Phys. Rev. Lett.}\ }\textbf {\bibinfo {volume} {125}},\ \bibinfo {pages} {030505} (\bibinfo {year} {2020})},\ \Eprint {https://arxiv.org/abs/1903.05124} {arXiv:1903.05124 [quant-ph]} \BibitemShut {NoStop}%
\bibitem [{\citenamefont {Li}\ \emph {et~al.}(2018)\citenamefont {Li}, \citenamefont {Chen},\ and\ \citenamefont {Fisher}}]{Li_2018}%
  \BibitemOpen
  \bibfield  {author} {\bibinfo {author} {\bibfnamefont {Y.}~\bibnamefont {Li}}, \bibinfo {author} {\bibfnamefont {X.}~\bibnamefont {Chen}},\ and\ \bibinfo {author} {\bibfnamefont {M.~P.~A.}\ \bibnamefont {Fisher}},\ }\bibfield  {title} {\bibinfo {title} {Quantum zeno effect and the many-body entanglement transition},\ }\href {https://doi.org/10.1103/PhysRevB.98.205136} {\bibfield  {journal} {\bibinfo  {journal} {Phys. Rev. B}\ }\textbf {\bibinfo {volume} {98}},\ \bibinfo {pages} {205136} (\bibinfo {year} {2018})}\BibitemShut {NoStop}%
\bibitem [{\citenamefont {Li}\ \emph {et~al.}(2019)\citenamefont {Li}, \citenamefont {Chen},\ and\ \citenamefont {Fisher}}]{Li_2019}%
  \BibitemOpen
  \bibfield  {author} {\bibinfo {author} {\bibfnamefont {Y.}~\bibnamefont {Li}}, \bibinfo {author} {\bibfnamefont {X.}~\bibnamefont {Chen}},\ and\ \bibinfo {author} {\bibfnamefont {M.~P.~A.}\ \bibnamefont {Fisher}},\ }\bibfield  {title} {\bibinfo {title} {Measurement-driven entanglement transition in hybrid quantum circuits},\ }\href {https://doi.org/10.1103/PhysRevB.100.134306} {\bibfield  {journal} {\bibinfo  {journal} {Phys. Rev. B}\ }\textbf {\bibinfo {volume} {100}},\ \bibinfo {pages} {134306} (\bibinfo {year} {2019})}\BibitemShut {NoStop}%
\bibitem [{\citenamefont {Skinner}\ \emph {et~al.}(2019)\citenamefont {Skinner}, \citenamefont {Ruhman},\ and\ \citenamefont {Nahum}}]{Skinner:2018tjl}%
  \BibitemOpen
  \bibfield  {author} {\bibinfo {author} {\bibfnamefont {B.}~\bibnamefont {Skinner}}, \bibinfo {author} {\bibfnamefont {J.}~\bibnamefont {Ruhman}},\ and\ \bibinfo {author} {\bibfnamefont {A.}~\bibnamefont {Nahum}},\ }\bibfield  {title} {\bibinfo {title} {{Measurement-Induced Phase Transitions in the Dynamics of Entanglement}},\ }\href {https://doi.org/10.1103/PhysRevX.9.031009} {\bibfield  {journal} {\bibinfo  {journal} {Phys. Rev. X}\ }\textbf {\bibinfo {volume} {9}},\ \bibinfo {pages} {031009} (\bibinfo {year} {2019})},\ \Eprint {https://arxiv.org/abs/1808.05953} {arXiv:1808.05953 [cond-mat.stat-mech]} \BibitemShut {NoStop}%
\bibitem [{\citenamefont {Cao}\ \emph {et~al.}(2019)\citenamefont {Cao}, \citenamefont {Tilloy},\ and\ \citenamefont {Luca}}]{Cao_2019}%
  \BibitemOpen
  \bibfield  {author} {\bibinfo {author} {\bibfnamefont {X.}~\bibnamefont {Cao}}, \bibinfo {author} {\bibfnamefont {A.}~\bibnamefont {Tilloy}},\ and\ \bibinfo {author} {\bibfnamefont {A.~D.}\ \bibnamefont {Luca}},\ }\bibfield  {title} {\bibinfo {title} {{Entanglement in a fermion chain under continuous monitoring}},\ }\href {https://doi.org/10.21468/SciPostPhys.7.2.024} {\bibfield  {journal} {\bibinfo  {journal} {SciPost Phys.}\ }\textbf {\bibinfo {volume} {7}},\ \bibinfo {pages} {024} (\bibinfo {year} {2019})}\BibitemShut {NoStop}%
\bibitem [{\citenamefont {Walter}\ \emph {et~al.}(2016)\citenamefont {Walter}, \citenamefont {Gross},\ and\ \citenamefont {Eisert}}]{Walter_2016}%
  \BibitemOpen
  \bibfield  {author} {\bibinfo {author} {\bibfnamefont {M.}~\bibnamefont {Walter}}, \bibinfo {author} {\bibfnamefont {D.}~\bibnamefont {Gross}},\ and\ \bibinfo {author} {\bibfnamefont {J.}~\bibnamefont {Eisert}},\ }\bibinfo {title} {Multipartite entanglement},\ in\ \href {https://doi.org/https://doi.org/10.1002/9783527805785.ch14} {\emph {\bibinfo {booktitle} {Quantum Information}}}\ (\bibinfo  {publisher} {John Wiley and Sons, Ltd},\ \bibinfo {year} {2016})\ Chap.~\bibinfo {chapter} {14}, pp.\ \bibinfo {pages} {293--330}\BibitemShut {NoStop}%
\bibitem [{\citenamefont {Vidal}\ and\ \citenamefont {Werner}(2002)}]{Vidal_2002}%
  \BibitemOpen
  \bibfield  {author} {\bibinfo {author} {\bibfnamefont {G.}~\bibnamefont {Vidal}}\ and\ \bibinfo {author} {\bibfnamefont {R.~F.}\ \bibnamefont {Werner}},\ }\bibfield  {title} {\bibinfo {title} {Computable measure of entanglement},\ }\bibfield  {journal} {\bibinfo  {journal} {Physical Review A}\ }\textbf {\bibinfo {volume} {65}},\ \href {https://doi.org/10.1103/physreva.65.032314} {10.1103/physreva.65.032314} (\bibinfo {year} {2002})\BibitemShut {NoStop}%
\bibitem [{\citenamefont {Hayden}\ \emph {et~al.}(2021)\citenamefont {Hayden}, \citenamefont {Parrikar},\ and\ \citenamefont {Sorce}}]{Hayden_2021}%
  \BibitemOpen
  \bibfield  {author} {\bibinfo {author} {\bibfnamefont {P.}~\bibnamefont {Hayden}}, \bibinfo {author} {\bibfnamefont {O.}~\bibnamefont {Parrikar}},\ and\ \bibinfo {author} {\bibfnamefont {J.}~\bibnamefont {Sorce}},\ }\bibfield  {title} {\bibinfo {title} {The markov gap for geometric reflected entropy},\ }\bibfield  {journal} {\bibinfo  {journal} {Journal of High Energy Physics}\ }\textbf {\bibinfo {volume} {2021}},\ \href {https://doi.org/10.1007/jhep10(2021)047} {10.1007/jhep10(2021)047} (\bibinfo {year} {2021})\BibitemShut {NoStop}%
\bibitem [{\citenamefont {Marcovitch}\ \emph {et~al.}(2009)\citenamefont {Marcovitch}, \citenamefont {Retzker}, \citenamefont {Plenio},\ and\ \citenamefont {Reznik}}]{Marcovitch_2009}%
  \BibitemOpen
  \bibfield  {author} {\bibinfo {author} {\bibfnamefont {S.}~\bibnamefont {Marcovitch}}, \bibinfo {author} {\bibfnamefont {A.}~\bibnamefont {Retzker}}, \bibinfo {author} {\bibfnamefont {M.~B.}\ \bibnamefont {Plenio}},\ and\ \bibinfo {author} {\bibfnamefont {B.}~\bibnamefont {Reznik}},\ }\bibfield  {title} {\bibinfo {title} {Critical and noncritical long-range entanglement in klein-gordon fields},\ }\href {https://doi.org/10.1103/PhysRevA.80.012325} {\bibfield  {journal} {\bibinfo  {journal} {Phys. Rev. A}\ }\textbf {\bibinfo {volume} {80}},\ \bibinfo {pages} {012325} (\bibinfo {year} {2009})}\BibitemShut {NoStop}%
\bibitem [{\citenamefont {Bayat}\ \emph {et~al.}(2010)\citenamefont {Bayat}, \citenamefont {Sodano},\ and\ \citenamefont {Bose}}]{Bayat_2010}%
  \BibitemOpen
  \bibfield  {author} {\bibinfo {author} {\bibfnamefont {A.}~\bibnamefont {Bayat}}, \bibinfo {author} {\bibfnamefont {P.}~\bibnamefont {Sodano}},\ and\ \bibinfo {author} {\bibfnamefont {S.}~\bibnamefont {Bose}},\ }\bibfield  {title} {\bibinfo {title} {Negativity as the entanglement measure to probe the kondo regime in the spin-chain kondo model},\ }\href {https://doi.org/10.1103/PhysRevB.81.064429} {\bibfield  {journal} {\bibinfo  {journal} {Phys. Rev. B}\ }\textbf {\bibinfo {volume} {81}},\ \bibinfo {pages} {064429} (\bibinfo {year} {2010})}\BibitemShut {NoStop}%
\bibitem [{\citenamefont {Kusuki}\ \emph {et~al.}(2019)\citenamefont {Kusuki}, \citenamefont {Kudler-Flam},\ and\ \citenamefont {Ryu}}]{Kusuki_2019}%
  \BibitemOpen
  \bibfield  {author} {\bibinfo {author} {\bibfnamefont {Y.}~\bibnamefont {Kusuki}}, \bibinfo {author} {\bibfnamefont {J.}~\bibnamefont {Kudler-Flam}},\ and\ \bibinfo {author} {\bibfnamefont {S.}~\bibnamefont {Ryu}},\ }\bibfield  {title} {\bibinfo {title} {Derivation of holographic negativity in ${\mathrm{ads}}_{3}/{\mathrm{cft}}_{2}$},\ }\href {https://doi.org/10.1103/PhysRevLett.123.131603} {\bibfield  {journal} {\bibinfo  {journal} {Phys. Rev. Lett.}\ }\textbf {\bibinfo {volume} {123}},\ \bibinfo {pages} {131603} (\bibinfo {year} {2019})}\BibitemShut {NoStop}%
\bibitem [{\citenamefont {Zou}\ \emph {et~al.}(2021)\citenamefont {Zou}, \citenamefont {Siva}, \citenamefont {Soejima}, \citenamefont {Mong},\ and\ \citenamefont {Zaletel}}]{Zou_2021}%
  \BibitemOpen
  \bibfield  {author} {\bibinfo {author} {\bibfnamefont {Y.}~\bibnamefont {Zou}}, \bibinfo {author} {\bibfnamefont {K.}~\bibnamefont {Siva}}, \bibinfo {author} {\bibfnamefont {T.}~\bibnamefont {Soejima}}, \bibinfo {author} {\bibfnamefont {R.~S.~K.}\ \bibnamefont {Mong}},\ and\ \bibinfo {author} {\bibfnamefont {M.~P.}\ \bibnamefont {Zaletel}},\ }\bibfield  {title} {\bibinfo {title} {Universal tripartite entanglement in one-dimensional many-body systems},\ }\href {https://doi.org/10.1103/PhysRevLett.126.120501} {\bibfield  {journal} {\bibinfo  {journal} {Phys. Rev. Lett.}\ }\textbf {\bibinfo {volume} {126}},\ \bibinfo {pages} {120501} (\bibinfo {year} {2021})}\BibitemShut {NoStop}%
\bibitem [{\citenamefont {Siva}\ \emph {et~al.}(2022)\citenamefont {Siva}, \citenamefont {Zou}, \citenamefont {Soejima}, \citenamefont {Mong},\ and\ \citenamefont {Zaletel}}]{Siva_2022}%
  \BibitemOpen
  \bibfield  {author} {\bibinfo {author} {\bibfnamefont {K.}~\bibnamefont {Siva}}, \bibinfo {author} {\bibfnamefont {Y.}~\bibnamefont {Zou}}, \bibinfo {author} {\bibfnamefont {T.}~\bibnamefont {Soejima}}, \bibinfo {author} {\bibfnamefont {R.~S.~K.}\ \bibnamefont {Mong}},\ and\ \bibinfo {author} {\bibfnamefont {M.~P.}\ \bibnamefont {Zaletel}},\ }\bibfield  {title} {\bibinfo {title} {Universal tripartite entanglement signature of ungappable edge states},\ }\bibfield  {journal} {\bibinfo  {journal} {Physical Review B}\ }\textbf {\bibinfo {volume} {106}},\ \href {https://doi.org/10.1103/physrevb.106.l041107} {10.1103/physrevb.106.l041107} (\bibinfo {year} {2022})\BibitemShut {NoStop}%
\bibitem [{\citenamefont {Liu}\ \emph {et~al.}(2022)\citenamefont {Liu}, \citenamefont {Sohal}, \citenamefont {Kudler-Flam},\ and\ \citenamefont {Ryu}}]{Liu_2022}%
  \BibitemOpen
  \bibfield  {author} {\bibinfo {author} {\bibfnamefont {Y.}~\bibnamefont {Liu}}, \bibinfo {author} {\bibfnamefont {R.}~\bibnamefont {Sohal}}, \bibinfo {author} {\bibfnamefont {J.}~\bibnamefont {Kudler-Flam}},\ and\ \bibinfo {author} {\bibfnamefont {S.}~\bibnamefont {Ryu}},\ }\bibfield  {title} {\bibinfo {title} {Multipartitioning topological phases by vertex states and quantum entanglement},\ }\bibfield  {journal} {\bibinfo  {journal} {Physical Review B}\ }\textbf {\bibinfo {volume} {105}},\ \href {https://doi.org/10.1103/physrevb.105.115107} {10.1103/physrevb.105.115107} (\bibinfo {year} {2022})\BibitemShut {NoStop}%
\bibitem [{\citenamefont {Liu}\ \emph {et~al.}(2024)\citenamefont {Liu}, \citenamefont {Kusuki}, \citenamefont {Kudler-Flam}, \citenamefont {Sohal},\ and\ \citenamefont {Ryu}}]{Liu_2024}%
  \BibitemOpen
  \bibfield  {author} {\bibinfo {author} {\bibfnamefont {Y.}~\bibnamefont {Liu}}, \bibinfo {author} {\bibfnamefont {Y.}~\bibnamefont {Kusuki}}, \bibinfo {author} {\bibfnamefont {J.}~\bibnamefont {Kudler-Flam}}, \bibinfo {author} {\bibfnamefont {R.}~\bibnamefont {Sohal}},\ and\ \bibinfo {author} {\bibfnamefont {S.}~\bibnamefont {Ryu}},\ }\bibfield  {title} {\bibinfo {title} {Multipartite entanglement in two-dimensional chiral topological liquids},\ }\href {https://doi.org/10.1103/PhysRevB.109.085108} {\bibfield  {journal} {\bibinfo  {journal} {Phys. Rev. B}\ }\textbf {\bibinfo {volume} {109}},\ \bibinfo {pages} {085108} (\bibinfo {year} {2024})}\BibitemShut {NoStop}%
\bibitem [{\citenamefont {Berthiere}\ \emph {et~al.}(2023)\citenamefont {Berthiere}, \citenamefont {Chen},\ and\ \citenamefont {Chen}}]{Berthiere:2023bwn}%
  \BibitemOpen
  \bibfield  {author} {\bibinfo {author} {\bibfnamefont {C.}~\bibnamefont {Berthiere}}, \bibinfo {author} {\bibfnamefont {B.}~\bibnamefont {Chen}},\ and\ \bibinfo {author} {\bibfnamefont {H.}~\bibnamefont {Chen}},\ }\bibfield  {title} {\bibinfo {title} {{Reflected entropy and Markov gap in Lifshitz theories}},\ }\href {https://doi.org/10.1007/JHEP09(2023)160} {\bibfield  {journal} {\bibinfo  {journal} {JHEP}\ }\textbf {\bibinfo {volume} {09}},\ \bibinfo {pages} {160}},\ \Eprint {https://arxiv.org/abs/2307.12247} {arXiv:2307.12247 [hep-th]} \BibitemShut {NoStop}%
\bibitem [{\citenamefont {Sang}\ \emph {et~al.}(2021)\citenamefont {Sang}, \citenamefont {Li}, \citenamefont {Zhou}, \citenamefont {Chen}, \citenamefont {Hsieh},\ and\ \citenamefont {Fisher}}]{Sang_2021}%
  \BibitemOpen
  \bibfield  {author} {\bibinfo {author} {\bibfnamefont {S.}~\bibnamefont {Sang}}, \bibinfo {author} {\bibfnamefont {Y.}~\bibnamefont {Li}}, \bibinfo {author} {\bibfnamefont {T.}~\bibnamefont {Zhou}}, \bibinfo {author} {\bibfnamefont {X.}~\bibnamefont {Chen}}, \bibinfo {author} {\bibfnamefont {T.~H.}\ \bibnamefont {Hsieh}},\ and\ \bibinfo {author} {\bibfnamefont {M.~P.}\ \bibnamefont {Fisher}},\ }\bibfield  {title} {\bibinfo {title} {Entanglement negativity at measurement-induced criticality},\ }\href {https://doi.org/10.1103/PRXQuantum.2.030313} {\bibfield  {journal} {\bibinfo  {journal} {PRX Quantum}\ }\textbf {\bibinfo {volume} {2}},\ \bibinfo {pages} {030313} (\bibinfo {year} {2021})}\BibitemShut {NoStop}%
\bibitem [{\citenamefont {Avakian}\ \emph {et~al.}(2024)\citenamefont {Avakian}, \citenamefont {Pereg-Barnea},\ and\ \citenamefont {Witczak-Krempa}}]{avakian2024longrangemultipartiteentanglementnear}%
  \BibitemOpen
  \bibfield  {author} {\bibinfo {author} {\bibfnamefont {S.~J.}\ \bibnamefont {Avakian}}, \bibinfo {author} {\bibfnamefont {T.}~\bibnamefont {Pereg-Barnea}},\ and\ \bibinfo {author} {\bibfnamefont {W.}~\bibnamefont {Witczak-Krempa}},\ }\href {https://arxiv.org/abs/2404.16095} {\bibinfo {title} {Long-range multipartite entanglement near measurement-induced transitions}} (\bibinfo {year} {2024}),\ \Eprint {https://arxiv.org/abs/2404.16095} {arXiv:2404.16095 [quant-ph]} \BibitemShut {NoStop}%
\bibitem [{\citenamefont {Paviglianiti}\ and\ \citenamefont {Silva}(2023)}]{Pavi_2023}%
  \BibitemOpen
  \bibfield  {author} {\bibinfo {author} {\bibfnamefont {A.}~\bibnamefont {Paviglianiti}}\ and\ \bibinfo {author} {\bibfnamefont {A.}~\bibnamefont {Silva}},\ }\bibfield  {title} {\bibinfo {title} {Multipartite entanglement in the measurement-induced phase transition of the quantum ising chain},\ }\href {https://doi.org/10.1103/PhysRevB.108.184302} {\bibfield  {journal} {\bibinfo  {journal} {Phys. Rev. B}\ }\textbf {\bibinfo {volume} {108}},\ \bibinfo {pages} {184302} (\bibinfo {year} {2023})}\BibitemShut {NoStop}%
\bibitem [{\citenamefont {Nielsen}\ and\ \citenamefont {Chuang}(2010)}]{nielsen2010quantum}%
  \BibitemOpen
  \bibfield  {author} {\bibinfo {author} {\bibfnamefont {M.}~\bibnamefont {Nielsen}}\ and\ \bibinfo {author} {\bibfnamefont {I.}~\bibnamefont {Chuang}},\ }\href {https://books.google.ca/books?id=-s4DEy7o-a0C} {\emph {\bibinfo {title} {Quantum Computation and Quantum Information: 10th Anniversary Edition}}}\ (\bibinfo  {publisher} {Cambridge University Press},\ \bibinfo {year} {2010})\BibitemShut {NoStop}%
\bibitem [{\citenamefont {Petz}(1986)}]{Petz:1986tvy}%
  \BibitemOpen
  \bibfield  {author} {\bibinfo {author} {\bibfnamefont {D.}~\bibnamefont {Petz}},\ }\bibfield  {title} {\bibinfo {title} {{Sufficient subalgebras and the relative entropy of states of a von Neumann algebra}},\ }\href {https://doi.org/10.1007/BF01212345} {\bibfield  {journal} {\bibinfo  {journal} {Commun. Math. Phys.}\ }\textbf {\bibinfo {volume} {105}},\ \bibinfo {pages} {123} (\bibinfo {year} {1986})}\BibitemShut {NoStop}%
\bibitem [{\citenamefont {Petz}(1988)}]{Petz:1988usv}%
  \BibitemOpen
  \bibfield  {author} {\bibinfo {author} {\bibfnamefont {D.}~\bibnamefont {Petz}},\ }\bibfield  {title} {\bibinfo {title} {{SUFFICIENCY OF CHANNELS OVER VON NEUMANN ALGEBRAS}},\ }\href {https://doi.org/10.1093/qmath/39.1.97} {\bibfield  {journal} {\bibinfo  {journal} {Quart. J. Math. Oxford Ser.}\ }\textbf {\bibinfo {volume} {39}},\ \bibinfo {pages} {97} (\bibinfo {year} {1988})}\BibitemShut {NoStop}%
\bibitem [{\citenamefont {Junge}\ \emph {et~al.}(2018)\citenamefont {Junge}, \citenamefont {Renner}, \citenamefont {Sutter}, \citenamefont {Wilde},\ and\ \citenamefont {Winter}}]{Junge:2015lmb}%
  \BibitemOpen
  \bibfield  {author} {\bibinfo {author} {\bibfnamefont {M.}~\bibnamefont {Junge}}, \bibinfo {author} {\bibfnamefont {R.}~\bibnamefont {Renner}}, \bibinfo {author} {\bibfnamefont {D.}~\bibnamefont {Sutter}}, \bibinfo {author} {\bibfnamefont {M.~M.}\ \bibnamefont {Wilde}},\ and\ \bibinfo {author} {\bibfnamefont {A.}~\bibnamefont {Winter}},\ }\bibfield  {title} {\bibinfo {title} {{Universal Recovery Maps and Approximate Sufficiency of Quantum Relative Entropy}},\ }\href {https://doi.org/10.1007/s00023-018-0716-0} {\bibfield  {journal} {\bibinfo  {journal} {Annales Henri Poincare}\ }\textbf {\bibinfo {volume} {19}},\ \bibinfo {pages} {2955} (\bibinfo {year} {2018})},\ \Eprint {https://arxiv.org/abs/1509.07127} {arXiv:1509.07127 [quant-ph]} \BibitemShut {NoStop}%
\bibitem [{\citenamefont {Sang}\ and\ \citenamefont {Hsieh}(2024)}]{sang2024stabilitymixedstatequantumphases}%
  \BibitemOpen
  \bibfield  {author} {\bibinfo {author} {\bibfnamefont {S.}~\bibnamefont {Sang}}\ and\ \bibinfo {author} {\bibfnamefont {T.~H.}\ \bibnamefont {Hsieh}},\ }\href {https://arxiv.org/abs/2404.07251} {\bibinfo {title} {Stability of mixed-state quantum phases via finite markov length}} (\bibinfo {year} {2024}),\ \Eprint {https://arxiv.org/abs/2404.07251} {arXiv:2404.07251 [quant-ph]} \BibitemShut {NoStop}%
\bibitem [{\citenamefont {Li}\ \emph {et~al.}(2021)\citenamefont {Li}, \citenamefont {Chen}, \citenamefont {Ludwig},\ and\ \citenamefont {Fisher}}]{Li_2021v2}%
  \BibitemOpen
  \bibfield  {author} {\bibinfo {author} {\bibfnamefont {Y.}~\bibnamefont {Li}}, \bibinfo {author} {\bibfnamefont {X.}~\bibnamefont {Chen}}, \bibinfo {author} {\bibfnamefont {A.~W.~W.}\ \bibnamefont {Ludwig}},\ and\ \bibinfo {author} {\bibfnamefont {M.~P.~A.}\ \bibnamefont {Fisher}},\ }\bibfield  {title} {\bibinfo {title} {Conformal invariance and quantum nonlocality in critical hybrid circuits},\ }\href {https://doi.org/10.1103/PhysRevB.104.104305} {\bibfield  {journal} {\bibinfo  {journal} {Phys. Rev. B}\ }\textbf {\bibinfo {volume} {104}},\ \bibinfo {pages} {104305} (\bibinfo {year} {2021})}\BibitemShut {NoStop}%
\bibitem [{\citenamefont {Vardhan}\ \emph {et~al.}(2024)\citenamefont {Vardhan}, \citenamefont {Wei},\ and\ \citenamefont {Zou}}]{Vardhan:2023pnm}%
  \BibitemOpen
  \bibfield  {author} {\bibinfo {author} {\bibfnamefont {S.}~\bibnamefont {Vardhan}}, \bibinfo {author} {\bibfnamefont {A.~Y.}\ \bibnamefont {Wei}},\ and\ \bibinfo {author} {\bibfnamefont {Y.}~\bibnamefont {Zou}},\ }\bibfield  {title} {\bibinfo {title} {{Petz recovery from subsystems in conformal field theory}},\ }\href {https://doi.org/10.1007/JHEP03(2024)016} {\bibfield  {journal} {\bibinfo  {journal} {JHEP}\ }\textbf {\bibinfo {volume} {03}},\ \bibinfo {pages} {016}},\ \Eprint {https://arxiv.org/abs/2307.14434} {arXiv:2307.14434 [hep-th]} \BibitemShut {NoStop}%
\bibitem [{\citenamefont {Wilde}(2015)}]{Wilde_2015}%
  \BibitemOpen
  \bibfield  {author} {\bibinfo {author} {\bibfnamefont {M.~M.}\ \bibnamefont {Wilde}},\ }\bibfield  {title} {\bibinfo {title} {Recoverability in quantum information theory},\ }\href {https://doi.org/10.1098/rspa.2015.0338} {\bibfield  {journal} {\bibinfo  {journal} {Proceedings of the Royal Society A: Mathematical, Physical and Engineering Sciences}\ }\textbf {\bibinfo {volume} {471}},\ \bibinfo {pages} {20150338} (\bibinfo {year} {2015})}\BibitemShut {NoStop}%
\bibitem [{\citenamefont {Gottesman}(1997)}]{gottesman1997stabilizercodesquantumerror}%
  \BibitemOpen
  \bibfield  {author} {\bibinfo {author} {\bibfnamefont {D.}~\bibnamefont {Gottesman}},\ }\href {https://arxiv.org/abs/quant-ph/9705052} {\bibinfo {title} {Stabilizer codes and quantum error correction}} (\bibinfo {year} {1997}),\ \Eprint {https://arxiv.org/abs/quant-ph/9705052} {arXiv:quant-ph/9705052 [quant-ph]} \BibitemShut {NoStop}%
\bibitem [{\citenamefont {Zabalo}\ \emph {et~al.}(2020)\citenamefont {Zabalo}, \citenamefont {Gullans}, \citenamefont {Wilson}, \citenamefont {Gopalakrishnan}, \citenamefont {Huse},\ and\ \citenamefont {Pixley}}]{Zabalo:2019sfl}%
  \BibitemOpen
  \bibfield  {author} {\bibinfo {author} {\bibfnamefont {A.}~\bibnamefont {Zabalo}}, \bibinfo {author} {\bibfnamefont {M.~J.}\ \bibnamefont {Gullans}}, \bibinfo {author} {\bibfnamefont {J.~H.}\ \bibnamefont {Wilson}}, \bibinfo {author} {\bibfnamefont {S.}~\bibnamefont {Gopalakrishnan}}, \bibinfo {author} {\bibfnamefont {D.~A.}\ \bibnamefont {Huse}},\ and\ \bibinfo {author} {\bibfnamefont {J.~H.}\ \bibnamefont {Pixley}},\ }\bibfield  {title} {\bibinfo {title} {{Critical properties of the measurement-induced transition in random quantum circuits}},\ }\href {https://doi.org/10.1103/PhysRevB.101.060301} {\bibfield  {journal} {\bibinfo  {journal} {Phys. Rev. B}\ }\textbf {\bibinfo {volume} {101}},\ \bibinfo {pages} {060301} (\bibinfo {year} {2020})},\ \Eprint {https://arxiv.org/abs/1911.00008} {arXiv:1911.00008 [cond-mat.dis-nn]} \BibitemShut {NoStop}%
\bibitem [{\citenamefont {Agrawal}\ \emph {et~al.}(2022)\citenamefont {Agrawal}, \citenamefont {Zabalo}, \citenamefont {Chen}, \citenamefont {Wilson}, \citenamefont {Potter}, \citenamefont {Pixley}, \citenamefont {Gopalakrishnan},\ and\ \citenamefont {Vasseur}}]{Agrawal:2021ukw}%
  \BibitemOpen
  \bibfield  {author} {\bibinfo {author} {\bibfnamefont {U.}~\bibnamefont {Agrawal}}, \bibinfo {author} {\bibfnamefont {A.}~\bibnamefont {Zabalo}}, \bibinfo {author} {\bibfnamefont {K.}~\bibnamefont {Chen}}, \bibinfo {author} {\bibfnamefont {J.~H.}\ \bibnamefont {Wilson}}, \bibinfo {author} {\bibfnamefont {A.~C.}\ \bibnamefont {Potter}}, \bibinfo {author} {\bibfnamefont {J.~H.}\ \bibnamefont {Pixley}}, \bibinfo {author} {\bibfnamefont {S.}~\bibnamefont {Gopalakrishnan}},\ and\ \bibinfo {author} {\bibfnamefont {R.}~\bibnamefont {Vasseur}},\ }\bibfield  {title} {\bibinfo {title} {{Entanglement and Charge-Sharpening Transitions in U(1) Symmetric Monitored Quantum Circuits}},\ }\href {https://doi.org/10.1103/PhysRevX.12.041002} {\bibfield  {journal} {\bibinfo  {journal} {Phys. Rev. X}\ }\textbf {\bibinfo {volume} {12}},\ \bibinfo {pages} {041002} (\bibinfo {year} {2022})},\ \Eprint {https://arxiv.org/abs/2107.10279} {arXiv:2107.10279 [cond-mat.dis-nn]} \BibitemShut {NoStop}%
\bibitem [{\citenamefont {Bao}\ \emph {et~al.}(2020)\citenamefont {Bao}, \citenamefont {Choi},\ and\ \citenamefont {Altman}}]{Bao_2020}%
  \BibitemOpen
  \bibfield  {author} {\bibinfo {author} {\bibfnamefont {Y.}~\bibnamefont {Bao}}, \bibinfo {author} {\bibfnamefont {S.}~\bibnamefont {Choi}},\ and\ \bibinfo {author} {\bibfnamefont {E.}~\bibnamefont {Altman}},\ }\bibfield  {title} {\bibinfo {title} {Theory of the phase transition in random unitary circuits with measurements},\ }\href {https://doi.org/10.1103/PhysRevB.101.104301} {\bibfield  {journal} {\bibinfo  {journal} {Phys. Rev. B}\ }\textbf {\bibinfo {volume} {101}},\ \bibinfo {pages} {104301} (\bibinfo {year} {2020})}\BibitemShut {NoStop}%
\bibitem [{\citenamefont {Marshakov}\ \emph {et~al.}(2010)\citenamefont {Marshakov}, \citenamefont {Mironov},\ and\ \citenamefont {Morozov}}]{Marshakov_2010}%
  \BibitemOpen
  \bibfield  {author} {\bibinfo {author} {\bibfnamefont {A.~V.}\ \bibnamefont {Marshakov}}, \bibinfo {author} {\bibfnamefont {A.~D.}\ \bibnamefont {Mironov}},\ and\ \bibinfo {author} {\bibfnamefont {A.~Y.}\ \bibnamefont {Morozov}},\ }\bibfield  {title} {\bibinfo {title} {Combinatorial expansions of conformal blocks},\ }\href {https://doi.org/10.1007/s11232-010-0067-6} {\bibfield  {journal} {\bibinfo  {journal} {Theoretical and Mathematical Physics}\ }\textbf {\bibinfo {volume} {164}},\ \bibinfo {pages} {831–852} (\bibinfo {year} {2010})}\BibitemShut {NoStop}%
\bibitem [{\citenamefont {Li}\ \emph {et~al.}(2023)\citenamefont {Li}, \citenamefont {Vijay},\ and\ \citenamefont {Fisher}}]{Li_KPZ2023}%
  \BibitemOpen
  \bibfield  {author} {\bibinfo {author} {\bibfnamefont {Y.}~\bibnamefont {Li}}, \bibinfo {author} {\bibfnamefont {S.}~\bibnamefont {Vijay}},\ and\ \bibinfo {author} {\bibfnamefont {M.~P.}\ \bibnamefont {Fisher}},\ }\bibfield  {title} {\bibinfo {title} {Entanglement domain walls in monitored quantum circuits and the directed polymer in a random environment},\ }\href {https://doi.org/10.1103/PRXQuantum.4.010331} {\bibfield  {journal} {\bibinfo  {journal} {PRX Quantum}\ }\textbf {\bibinfo {volume} {4}},\ \bibinfo {pages} {010331} (\bibinfo {year} {2023})}\BibitemShut {NoStop}%
\bibitem [{\citenamefont {Swingle}\ and\ \citenamefont {Wang}(2019)}]{Swingle:2018dto}%
  \BibitemOpen
  \bibfield  {author} {\bibinfo {author} {\bibfnamefont {B.~G.}\ \bibnamefont {Swingle}}\ and\ \bibinfo {author} {\bibfnamefont {Y.}~\bibnamefont {Wang}},\ }\bibfield  {title} {\bibinfo {title} {{Recovery map for fermionic Gaussian channels}},\ }\href {https://doi.org/10.1063/1.5093326} {\bibfield  {journal} {\bibinfo  {journal} {J. Math. Phys.}\ }\textbf {\bibinfo {volume} {60}},\ \bibinfo {pages} {072202} (\bibinfo {year} {2019})},\ \Eprint {https://arxiv.org/abs/1811.04956} {arXiv:1811.04956 [quant-ph]} \BibitemShut {NoStop}%
\bibitem [{\citenamefont {Lami}\ \emph {et~al.}(2018)\citenamefont {Lami}, \citenamefont {Das},\ and\ \citenamefont {Wilde}}]{Lami_2018}%
  \BibitemOpen
  \bibfield  {author} {\bibinfo {author} {\bibfnamefont {L.}~\bibnamefont {Lami}}, \bibinfo {author} {\bibfnamefont {S.}~\bibnamefont {Das}},\ and\ \bibinfo {author} {\bibfnamefont {M.~M.}\ \bibnamefont {Wilde}},\ }\bibfield  {title} {\bibinfo {title} {Approximate reversal of quantum gaussian dynamics},\ }\href {https://doi.org/10.1088/1751-8121/aaad26} {\bibfield  {journal} {\bibinfo  {journal} {Journal of Physics A: Mathematical and Theoretical}\ }\textbf {\bibinfo {volume} {51}},\ \bibinfo {pages} {125301} (\bibinfo {year} {2018})}\BibitemShut {NoStop}%
\bibitem [{\citenamefont {Weinstein}\ \emph {et~al.}(2023)\citenamefont {Weinstein}, \citenamefont {Sajith}, \citenamefont {Altman},\ and\ \citenamefont {Garratt}}]{Weinstein_2023}%
  \BibitemOpen
  \bibfield  {author} {\bibinfo {author} {\bibfnamefont {Z.}~\bibnamefont {Weinstein}}, \bibinfo {author} {\bibfnamefont {R.}~\bibnamefont {Sajith}}, \bibinfo {author} {\bibfnamefont {E.}~\bibnamefont {Altman}},\ and\ \bibinfo {author} {\bibfnamefont {S.~J.}\ \bibnamefont {Garratt}},\ }\bibfield  {title} {\bibinfo {title} {Nonlocality and entanglement in measured critical quantum ising chains},\ }\bibfield  {journal} {\bibinfo  {journal} {Physical Review B}\ }\textbf {\bibinfo {volume} {107}},\ \href {https://doi.org/10.1103/physrevb.107.245132} {10.1103/physrevb.107.245132} (\bibinfo {year} {2023})\BibitemShut {NoStop}%
\bibitem [{\citenamefont {Kim}\ \emph {et~al.}(2022)\citenamefont {Kim}, \citenamefont {Shi}, \citenamefont {Kato},\ and\ \citenamefont {Albert}}]{Kim_2022}%
  \BibitemOpen
  \bibfield  {author} {\bibinfo {author} {\bibfnamefont {I.~H.}\ \bibnamefont {Kim}}, \bibinfo {author} {\bibfnamefont {B.}~\bibnamefont {Shi}}, \bibinfo {author} {\bibfnamefont {K.}~\bibnamefont {Kato}},\ and\ \bibinfo {author} {\bibfnamefont {V.~V.}\ \bibnamefont {Albert}},\ }\bibfield  {title} {\bibinfo {title} {Modular commutator in gapped quantum many-body systems},\ }\bibfield  {journal} {\bibinfo  {journal} {Physical Review B}\ }\textbf {\bibinfo {volume} {106}},\ \href {https://doi.org/10.1103/physrevb.106.075147} {10.1103/physrevb.106.075147} (\bibinfo {year} {2022})\BibitemShut {NoStop}%
\bibitem [{\citenamefont {Fan}(2022)}]{Fan_2022}%
  \BibitemOpen
  \bibfield  {author} {\bibinfo {author} {\bibfnamefont {R.}~\bibnamefont {Fan}},\ }\bibfield  {title} {\bibinfo {title} {From entanglement generated dynamics to the gravitational anomaly and chiral central charge},\ }\href {https://doi.org/10.1103/PhysRevLett.129.260403} {\bibfield  {journal} {\bibinfo  {journal} {Phys. Rev. Lett.}\ }\textbf {\bibinfo {volume} {129}},\ \bibinfo {pages} {260403} (\bibinfo {year} {2022})}\BibitemShut {NoStop}%
\bibitem [{\citenamefont {Zou}\ \emph {et~al.}(2022)\citenamefont {Zou}, \citenamefont {Shi}, \citenamefont {Sorce}, \citenamefont {Lim},\ and\ \citenamefont {Kim}}]{Zou_2022v2}%
  \BibitemOpen
  \bibfield  {author} {\bibinfo {author} {\bibfnamefont {Y.}~\bibnamefont {Zou}}, \bibinfo {author} {\bibfnamefont {B.}~\bibnamefont {Shi}}, \bibinfo {author} {\bibfnamefont {J.}~\bibnamefont {Sorce}}, \bibinfo {author} {\bibfnamefont {I.~T.}\ \bibnamefont {Lim}},\ and\ \bibinfo {author} {\bibfnamefont {I.~H.}\ \bibnamefont {Kim}},\ }\bibfield  {title} {\bibinfo {title} {Modular commutators in conformal field theory},\ }\bibfield  {journal} {\bibinfo  {journal} {Physical Review Letters}\ }\textbf {\bibinfo {volume} {129}},\ \href {https://doi.org/10.1103/physrevlett.129.260402} {10.1103/physrevlett.129.260402} (\bibinfo {year} {2022})\BibitemShut {NoStop}%
\bibitem [{\citenamefont {Gass}\ and\ \citenamefont {Levin}(2024)}]{gass2024manybodysystemsspuriousmodular}%
  \BibitemOpen
  \bibfield  {author} {\bibinfo {author} {\bibfnamefont {J.}~\bibnamefont {Gass}}\ and\ \bibinfo {author} {\bibfnamefont {M.}~\bibnamefont {Levin}},\ }\href {https://arxiv.org/abs/2405.15892} {\bibinfo {title} {Many-body systems with spurious modular commutators}} (\bibinfo {year} {2024}),\ \Eprint {https://arxiv.org/abs/2405.15892} {arXiv:2405.15892 [quant-ph]} \BibitemShut {NoStop}%
\bibitem [{\citenamefont {Qi}\ \emph {et~al.}(2012)\citenamefont {Qi}, \citenamefont {Katsura},\ and\ \citenamefont {Ludwig}}]{Qi_2012}%
  \BibitemOpen
  \bibfield  {author} {\bibinfo {author} {\bibfnamefont {X.-L.}\ \bibnamefont {Qi}}, \bibinfo {author} {\bibfnamefont {H.}~\bibnamefont {Katsura}},\ and\ \bibinfo {author} {\bibfnamefont {A.~W.~W.}\ \bibnamefont {Ludwig}},\ }\bibfield  {title} {\bibinfo {title} {General relationship between the entanglement spectrum and the edge state spectrum of topological quantum states},\ }\href {https://doi.org/10.1103/PhysRevLett.108.196402} {\bibfield  {journal} {\bibinfo  {journal} {Phys. Rev. Lett.}\ }\textbf {\bibinfo {volume} {108}},\ \bibinfo {pages} {196402} (\bibinfo {year} {2012})}\BibitemShut {NoStop}%
\bibitem [{\citenamefont {Shi}\ \emph {et~al.}(2020)\citenamefont {Shi}, \citenamefont {Kato},\ and\ \citenamefont {Kim}}]{Shi_2020}%
  \BibitemOpen
  \bibfield  {author} {\bibinfo {author} {\bibfnamefont {B.}~\bibnamefont {Shi}}, \bibinfo {author} {\bibfnamefont {K.}~\bibnamefont {Kato}},\ and\ \bibinfo {author} {\bibfnamefont {I.~H.}\ \bibnamefont {Kim}},\ }\bibfield  {title} {\bibinfo {title} {Fusion rules from entanglement},\ }\href {https://doi.org/10.1016/j.aop.2020.168164} {\bibfield  {journal} {\bibinfo  {journal} {Annals of Physics}\ }\textbf {\bibinfo {volume} {418}},\ \bibinfo {pages} {168164} (\bibinfo {year} {2020})}\BibitemShut {NoStop}%
\bibitem [{\citenamefont {Williamson}\ \emph {et~al.}(2019)\citenamefont {Williamson}, \citenamefont {Dua},\ and\ \citenamefont {Cheng}}]{Williamson_2019}%
  \BibitemOpen
  \bibfield  {author} {\bibinfo {author} {\bibfnamefont {D.~J.}\ \bibnamefont {Williamson}}, \bibinfo {author} {\bibfnamefont {A.}~\bibnamefont {Dua}},\ and\ \bibinfo {author} {\bibfnamefont {M.}~\bibnamefont {Cheng}},\ }\bibfield  {title} {\bibinfo {title} {Spurious topological entanglement entropy from subsystem symmetries},\ }\bibfield  {journal} {\bibinfo  {journal} {Physical Review Letters}\ }\textbf {\bibinfo {volume} {122}},\ \href {https://doi.org/10.1103/physrevlett.122.140506} {10.1103/physrevlett.122.140506} (\bibinfo {year} {2019})\BibitemShut {NoStop}%
\bibitem [{\citenamefont {Kim}\ \emph {et~al.}(2023)\citenamefont {Kim}, \citenamefont {Levin}, \citenamefont {Lin}, \citenamefont {Ranard},\ and\ \citenamefont {Shi}}]{Kim_2023}%
  \BibitemOpen
  \bibfield  {author} {\bibinfo {author} {\bibfnamefont {I.~H.}\ \bibnamefont {Kim}}, \bibinfo {author} {\bibfnamefont {M.}~\bibnamefont {Levin}}, \bibinfo {author} {\bibfnamefont {T.-C.}\ \bibnamefont {Lin}}, \bibinfo {author} {\bibfnamefont {D.}~\bibnamefont {Ranard}},\ and\ \bibinfo {author} {\bibfnamefont {B.}~\bibnamefont {Shi}},\ }\bibfield  {title} {\bibinfo {title} {Universal lower bound on topological entanglement entropy},\ }\href {https://doi.org/10.1103/PhysRevLett.131.166601} {\bibfield  {journal} {\bibinfo  {journal} {Phys. Rev. Lett.}\ }\textbf {\bibinfo {volume} {131}},\ \bibinfo {pages} {166601} (\bibinfo {year} {2023})}\BibitemShut {NoStop}%
\bibitem [{\citenamefont {Bravyi}(2004)}]{bravyi2004lagrangian}%
  \BibitemOpen
  \bibfield  {author} {\bibinfo {author} {\bibfnamefont {S.}~\bibnamefont {Bravyi}},\ }\bibfield  {title} {\bibinfo {title} {Lagrangian representation for fermionic linear optics},\ }\href@noop {} {\bibfield  {journal} {\bibinfo  {journal} {arXiv preprint quant-ph/0404180}\ } (\bibinfo {year} {2004})}\BibitemShut {NoStop}%
\end{thebibliography}%
\end{document}